\pdfoutput=1
\documentclass[11pt]{article}

\usepackage{fullpage}
\usepackage{amsmath,amssymb,graphicx,epsfig,latexsym,verbatim,url}
\usepackage{amsthm}
\usepackage{bbm}
\usepackage[numbers]{natbib}
\usepackage{algorithm}
\usepackage{algorithmicx}
\usepackage{algpseudocode}
\usepackage{subcaption}
\usepackage[colorlinks=true,citecolor=red]{hyperref}
\graphicspath{{figures/}}

\usepackage{amsbsy}

\newtheorem{theo}{Theorem}

\newtheorem{lemm}{Lemma}

\newtheorem{defi}{Definition}

\usepackage{setspace}

\newcommand{\myexp}[1]{\exp{(#1)}}
\usepackage{mathtools}
\DeclarePairedDelimiter{\ceil}{\lceil}{\rceil}
\usepackage{color}

\newcommand{\loadshare}{\alpha}

\title{Sprinklers: A Randomized Variable-Size Striping Approach to Reordering-Free Load-Balanced Switching}
\author{
Weijun Ding \footnote{School of Industrial and Systems Engineering, Georgia Institute of Technology, wding34@gatech.edu} 
\and 
Jim Xu \footnote{College of Computing, Georgia Institute of Technology, jx@cc.gatech.edu}
\and 
Jim Dai\footnote{School of Operations Research and Information
  Engineering, Cornell University, Ithaca, NY 14853; on leave from
  Georgia Institute of Technology,
jim.dai@cornell.edu} 
\and 
Yang Song\footnote{Electrical and Computer Engineering, University of California, San Diego, y6song@eng.ucsd.edu} 
\and 
Bill Lin\footnote{Electrical and Computer Engineering, University of California, San Diego, billlin@eng.ucsd.edu}
}

\begin{document}
\date{\normalsize June 7, 2014}
\maketitle

\begin{abstract}
Internet traffic continues to grow exponentially,
calling for switches that can scale well in both size and speed.
While load-balanced switches can achieve such scalability,
they suffer from a fundamental packet reordering problem.
Existing proposals 
either suffer from poor worst-case packet delays or require sophisticated matching mechanisms.
In this paper, we propose a new family of stable load-balanced switches called ``Sprinklers'' that has comparable implementation cost and performance as the baseline load-balanced switch, but yet can guarantee packet ordering.
The main idea is to force all packets within the same virtual output queue (VOQ) to traverse the same ``fat path'' through the switch, so that packet 
reordering cannot occur.
At the core of Sprinklers are two key innovations: 
a randomized way to determine the ``fat path'' for each VOQ, and
a way to determine its ``fatness'' roughly in proportion to the rate of the VOQ.
These innovations enable Sprinklers to achieve near-perfect load-balancing
\emph{under arbitrary admissible traffic}.
Proving this property rigorously using novel worst-case large deviation techniques is 
another key 
contribution of this work.
\end{abstract}

\section{Introduction}
\label{sec:intro}

Internet service providers need high-performance switch architectures that can scale well in both size and speed, provide throughput guarantees, achieve low latency, and maintain packet ordering.
However, conventional crossbar-based switch architectures with centralized scheduling and 
arbitrary per-packet dynamic switch configurations are not scalable.

An alternative class of switch architecture is the load-balanced switch,
first introduced by Chang et al.~\cite{chang2002load,chang2002load2},
and later further developed by
others (e.g.~\cite{jaramillo2008padded,keslassy2004load,lin2010concurrent}).
These architectures rely on two switching stages for routing packets.
Fig.~\ref{fig:load balance} shows a diagram of a generic
two-stage load-balanced switch. The first switch connects the first
stage of input ports to the center stage of intermediate ports,
and the second switch connects the center stage of intermediate ports
to the final stage of output ports. Both switching stages execute a deterministic
connection pattern such that each input to a switching stage is connected to each
output of the switch at 1/$N$th of the time.  This can be implemented for example
using a deterministic round-robin switch (see Sec.~\ref{sec: lsf scheduling}).  Alternatively, as shown in~\cite{keslassy2004load},
the deterministic connection pattern can also be efficiently implemented using optics in
which all inputs are connected to all outputs of a switching stage in parallel
at a rate of 1/$N$th the line rate.
This class of architectures appears to be a practical way to scale high-performance switches to
very high capacities and line rates. 

\begin{figure}[htb]
  \centering
    \includegraphics[width=0.45\textwidth]{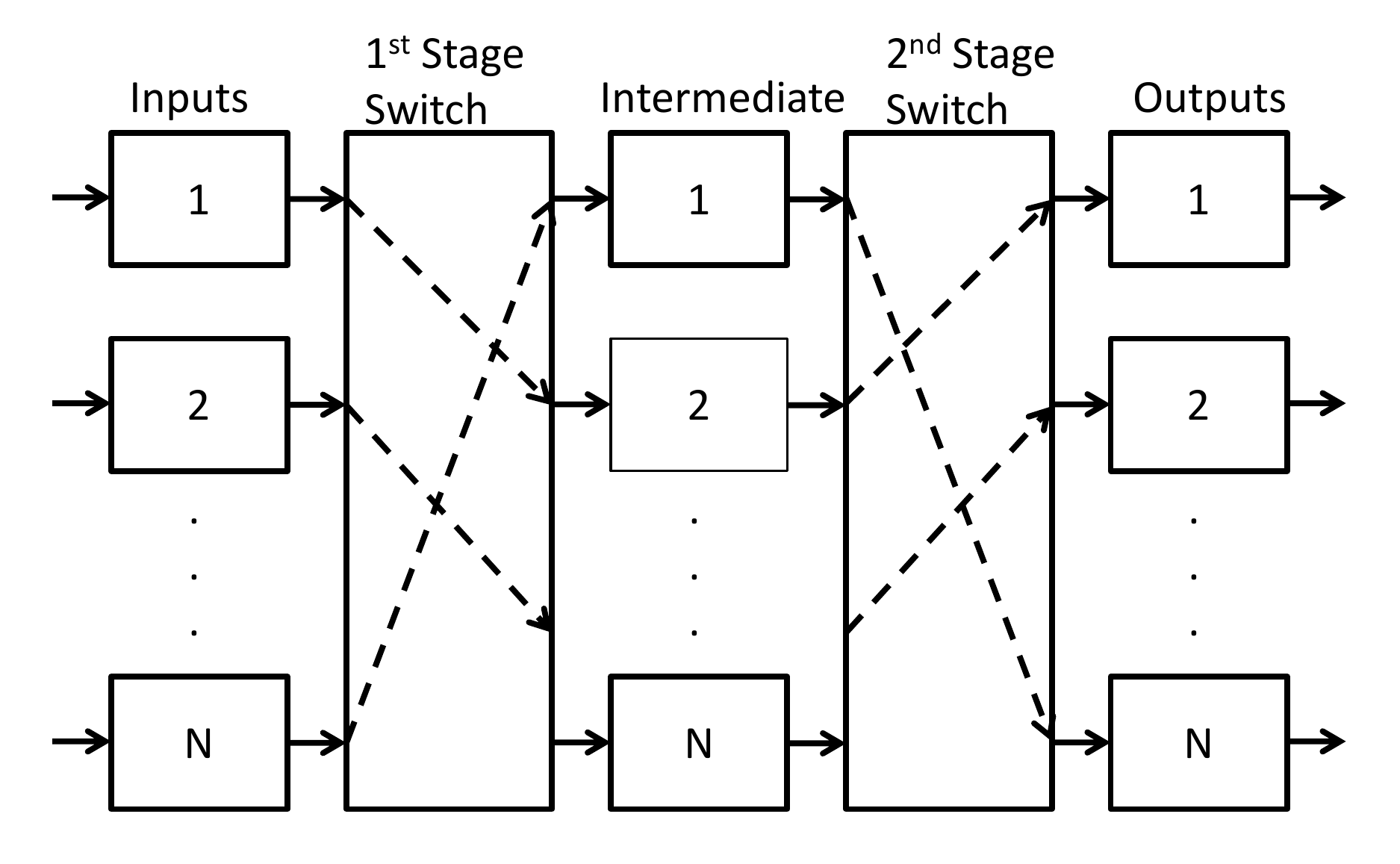}
    \caption{Generic load-balanced switch.}
\label{fig:load balance}
\end{figure}

Although the baseline load-balanced switch originally proposed
in~\cite{chang2002load} is capable of achieving throughput guarantees, it
has the critical problem that packet departures can be badly 
out-of-order.
In the baseline load-balanced switch,
consecutive packets at an input port are spread to all $N$ intermediate ports upon arrival.
Packets going through different intermediate ports may encounter \emph{different queueing delays}.
Thus, some of these packets may arrive at their output ports 
out-of-order.
This is detrimental to Internet traffic since the
widely used TCP transport protocol falsely regards out-of-order
packets as indications of congestion and packet loss.
Therefore, a number of
researchers have explored this
packet ordering problem (e.g.~\cite{jaramillo2008padded,keslassy2004load,lin2010concurrent}).

\subsection{Our approach}

In this paper, we propose ``Sprinklers,'' a new load-balanced switching solution that has comparable implementation cost and computational complexity, and similar performance guarantees as the baseline load-balanced switch, yet can guarantee packet ordering.
The Sprinklers approach has a mild flavor of a simple,
yet flawed solution called ``TCP hashing''\footnote{TCP hashing is referred to as ``Application Flow Based Routing'' (AFBR) in~\cite{keslassy2004load}.}~\cite{keslassy2004load}.
The idea of TCP hashing is to force all packets at an input port that belong to the same application (TCP or UDP) flow to go through the same \emph{randomly chosen} intermediate port by hashing on the packet's flow identifier.
This approach ensures that all packets that belong to the same application flow will depart the switch in order because they will encounter the \emph{same queueing delay} through the same intermediate port.
Although the TCP hashing scheme is simple and intuitive, it cannot guarantee stability because an intermediate port can easily be oversubscribed by having too many \emph{large} flows randomly assigned to it.

Like the TCP hashing approach in which all packets within an application flow are forced to go through the same switching path,
a Sprinklers switch requires all packets within a Virtual Output Queue (VOQ)
to go down the \emph{same} ``fat path'' called a stripe.
Specifically, for each VOQ, a Sprinklers switch stripes packets within a VOQ across an interval of consecutive
intermediate ports, which we call a \emph{stripe interval}.
The number of intermediate ports $L$ in the stripe interval is referred to as its \emph{size}.
In the Sprinklers approach, the size $L$ of the stripe interval is determined roughly in proportion to the \emph{rate} of the VOQ,
and the placement of the $L$ consecutive intermediate ports is by means of \emph{random permutation}.
Therefore, the $N^2$ stripe intervals, corresponding to the $N^2$ VOQs at the $N$ input ports,
can be very different from one another in both sizes and placements.
The placement of variable-size stripes by means of random permutation ensures
that the traffic loads associated with them get evenly distributed across the $N$ intermediate ports.

For example, suppose a VOQ is assigned to the stripe interval $(8, 12]$ based on its arrival rate and random permutation,
which corresponds to intermediate ports $\{9, 10, 11, 12\}$.
Then the incoming packets belonging to this VOQ will be grouped in arrival order into stripes of $L = 4$ packets each.
Once a stripe is filled, the $L = 4$ packets in the stripe will be transmitted to intermediate ports
$9$, $10$, $11$, and $12$, respectively.



By making a stripe the basic unit of scheduling at both input and intermediate ports, 
a Sprinklers switch ensures that every stripe of packets departs from
its input port and arrives at its output port both ``in one burst''
(i.e., in consecutive time slots).  This ``no interleaving between the servicing of two stripes" service guarantee, combined with 
the FCFS order in which packets within a stripe and stripes within a VOQ are served, ensures that 
packet reordering cannot happen within any VOQ, and hence cannot happen within any application flow either.

\subsection{Contributions of the paper}

This paper makes the following major contributions:
\begin{itemize}
\item First, we introduce the design of a new load-balanced switch architecture based on randomized and variable-size striping that we call Sprinklers.
Sprinklers is indeed scalable in that all its algorithmic aspects can be implemented in constant time at each input port and 
intermediate port in a fully distributed manner.

\item Second, we develop novel large deviation techniques to prove that Sprinklers is stable under intensive arrival rates with overwhelming probability, while guaranteeing packet order.


\item Third, to our knowledge, 
Sprinklers is the first load-balanced switch architecture based on randomization that can guarantee both stability and packet ordering, which we hope 
will become
the catalyst to a rich family of solutions based on the simple principles of 
randomization and variable stripe sizing.
\end{itemize}

\subsection{Outline of the paper}

The rest of the paper is organized as follows.
Sec.~\ref{sec:related} provides a brief review of existing load-balanced switch solutions.
Sec.~\ref{sec:scheme} introduces the structure of our switch and explains how it works.
Sec.~\ref{sec:stability} provides a rigorous proof of stability based on the use of
convex optimization theory and negative association theory to develop a Chernoff bound for the overload probability.
This section also presents our worst-case large deviation results.
Sec.~\ref{sec:simu} evaluates our proposed architecture, and
Sec.~\ref{sec:conc} concludes the paper.

\section{Existing Load-Balanced Switch Solutions}
\label{sec:related}

\subsection{Hashing-Based}

As discussed in Sec.~\ref{sec:intro}, packets belonging to the same TCP flow can be guaranteed to depart from their output port in order if they are forced to go through the same intermediate port.
The selection of intermediate port can be easily achieved by hashing on the packet header to obtain a value from 1 to $N$.
Despite its simplicity, the main drawback of this TCP hashing approach is that stability cannot be guaranteed~\cite{keslassy2004load}.

\subsection{Aggregation-Based}

An alternative class of algorithms to hashing is based on aggregation of packets into frames.
One approach called Uniform Frame Spreading (UFS)~\cite{keslassy2004load} prevents reordering by requiring that each input first accumulates a full-frame of $N$ packets, all going to the same output, before uniformly spreading the $N$ packets to the $N$ intermediate ports.
Packets are accumulated in separate virtual output queues (VOQs) at each input for storing packets in accordance to their output.
When a full-frame is available, the $N$ packets are spread by placing one packet at each of the $N$ intermediate ports.
This ensures that the lengths of the queues of packets destined to the same output are the same at every intermediate port,
which ensures every packet going to the same output experiences the same queuing delay independent of the path that it takes from input to output.
Although it has been shown in~\cite{keslassy2004load} that UFS achieves 100\% throughput for any admissible traffic pattern,
the main drawback of UFS is that it suffers from long delays, $O(N^3)$ delay in the worst-case,
due to the need to wait for a full-frame before transmission.
The performance of UFS is particularly bad at light loads because slow packet arrivals lead to much longer accumulation times.

An alternative aggregation-based algorithm that avoids the need to wait for a full-frame is called Full Ordered Frames First (FOFF)~\cite{keslassy2004load}.
As with UFS, FOFF maintains VOQs at each input.
Whenever possible, FOFF will serve full-frames first.
When there is no full-frame available, FOFF will serve the other queues in a round-robin manner.
However, when incomplete frames are served, packets can arrive at the output out of order.
It has been shown in~\cite{keslassy2004load} that the amount of reordering is always bounded by $O(N^2)$ with FOFF.
Therefore, FOFF adds a reordering buffer of size $O(N^2)$ at each output to ensure that packets depart in order.
It has been shown in~\cite{keslassy2004load} that FOFF achieves 100\% throughput for any admissible traffic pattern,
but the added reordering buffers lead to an $O(N^2)$ in packet delays.

Another aggregation-based algorithm called Padded Frames (PF)~\cite{jaramillo2008padded} was proposed to avoid the need to accumulate full-frames.
Like FOFF, whenever possible, FOFF will serve full-frames first.
When 
no full-frame is available, PF will search among its VOQ at each input to find the longest one.
If the length of the longest queue exceeds some threshold $T$, PF will pad the frame with
fake packets
to create a full-frame.
This full-frame of packets, including the fake packets, are uniformly spread across the $N$ intermediate ports, just like UFS.
It has been shown in~\cite{jaramillo2008padded} that PF achieves 100\% throughput for any admissible traffic pattern,
but its worst-case delay bound is still $O(N^3)$.

\subsection{Matching-Based}

Finally, packet ordering can be guaranteed in load-balanced switches via another approach called a Concurrent Matching Switch (CMS)~\cite{lin2010concurrent}.
Like hashing-based and aggregation-based load-balanced switch designs,
CMS is also a fully distributed solution.
However, instead of bounding the amount of packet reordering through the switch, or requiring packet aggregation,
a CMS enforces packet ordering throughout the switch by using a
fully distributed load-balanced scheduling approach.
Instead of load-balancing packets, a CMS load-balances request
tokens among intermediate ports, where each intermediate
port concurrently solves a local matching
problem based only on its local token count.
Then, each intermediate
port independently selects a VOQ from each input to serve, such that the packets
selected can traverse the two load-balanced switch stages without conflicts.
Packets from selected VOQs depart in order from the inputs, through the
intermediate ports, and finally through the outputs.
Each intermediate port has $N$ time slots to
perform each matching, so the complexity of existing matching
algorithms can be amortized by a factor of $N$.

\section{Our scheme}
\label{sec:scheme}

Sprinklers has the same architecture as the baseline load-balanced switch
(see Fig.~\ref{fig:load balance}), but it differs in the way that
it routes and schedules packets for service at the input
and intermediate ports. Also, Sprinklers have $N$ VOQs at each input port.
In this section, we first provide some intuition behind the Sprinklers approach,
followed by how the Sprinklers switch operates,
including the striping mechanism for routing packets through the switch
and the companion stripe scheduling policy.

\subsection{Intuition Behind the Sprinklers Approach}

The Sprinklers approach is based on three techniques for balancing traffic evenly across all $N$ intermediate ports:
permutation, randomization, and variable-size striping.
To provide some intuition as to why all three techniques are necessary, we use
an analogy from which the name Sprinklers is derived.
Consider the task of watering a lawn consisting of $N$ identically sized areas using $N$ sprinklers with different pressure.
The objective is to distribute an (ideally) identical amount of water to each area.  This
corresponds to evenly distributing traffic inside the $N$ VOQs (sprinklers) entering a certain input port to the $N$ intermediate
ports (lawn areas)  
under the above-mentioned constraint that all traffic inside a VOQ must go through the same 
set of intermediate ports (i.e., the stripe interval) to which the VOQ is mapped.


An intuitive and sensible first step is to aim exactly one sprinkler at each lawn area,
since aiming more than one sprinklers at one lawn area
clearly could lead to it being flooded.
In other words, the ``aiming function,'' or ``which sprinkler is aimed at which lawn area,'' 
is essentially a permutation over the set $\{1, 2, \cdots, N\}$.
This permutation alone however cannot do the trick, because the water pressure (traffic rate) is 
different from one sprinkler to another, and the lawn area aimed at by a sprinkler with high water 
pressure will surely be flooded.
To deal with such disparity in water pressures, we 
set the ``spray angle range'' of a sprinkler, which corresponds to the size (say $L$) of the stripe interval for the
corresponding VOQ, proportional to its water pressure, and evenly distribute this water pressure 
across the 
$L$ ``streams'' of water that go to the target lawn area and $L-1$ ``neighboring'' lawn areas.

However, such water pressure equalization (i.e., variable stripe sizing) alone does not prevent all scenarios of load-imbalance 
because it shuffles water around only ``locally."
For example, if a larger than average 
number of high pressure sprinklers are aimed at a cluster of lawn areas close to one another, some area within 
this cluster will be flooded.   Hence, 
a simple yet powerful randomized algorithm is brought in 
to shuffle water around globally: we simply sample this permutation 
at uniform random
from the set of all $N!$ permutations.



Besides load-balancing, there is another important reason for the size of a stripe interval to be set roughly 
proportional to the traffic rate of the corresponding VOQ.  In 
some existing solutions to the packet reordering
problem with load-balanced switches, each VOQ has to accumulate a full frame of $N$ packets before the frame can depart from the 
input port.  For a VOQ with low traffic rate, this buffering delay could be painfully long.  By adopting
rate-proportional
stripe sizing, a Sprinklers switch significantly reduces the buffering delays experienced by the low-rate VOQs.

As far as the Sprinklers analogy goes, the combination of randomized permutation and water pressure equalization,
with proper ``manifoldization'' (i.e., considering $N + 1$ as $1$),
will provably ensure that the amount of water going to each lawn area is very even with high 
probability.   However, because the servicing of any two stripes cannot interleave in a Sprinklers switch (to ensure correct
packet order), two stripes have to be serviced in two different frames ($N$ time slots),
even if their stripe intervals overlap only slightly.  
A rampant occurrence of such slight overlaps, which can happen if the stripe interval size of a VOQ is set strictly proportional
to the rate of the VOQ (rounded to an integer), 
will result in gross waste of service capacity, and significantly reduce the maximum achievable throughput.

Therefore, we would like any two stripe intervals to either ``bear hug"
(i.e., one contained entirely in the other) or does not touch (no overlap between the intervals) 
each other.  Our solution is a classical ``computer science'' one: making $N$ a power of 2 (very reasonable 
in the modern switching literature) and every stripe interval a dyadic one (resulting from dividing the 
whole interval (0, N] into $2^k$ equal-sized subintervals for an integer $k \le \log_2 N$).  
Now that the spray angle range of a sprinkler has 
to be a power of 2, the water pressure per stream could vary from one sprinkler to another by a maximum factor of 2.  
However, as shown later in Sec.~\ref{sec:stability},
strong statistical load-balancing guarantees can still be rigorously proven
despite such variations.

\subsection{Operations of the Sprinklers Switch}

As just explained, traffic inside each of the $N^2$ VOQs is switched through
a dyadic interval of intermediate ports that is just large enough to bring the load imposed by the VOQ on any intermediate port 
within this interval (i.e., ``water pressure per stream") below a certain threshold.  
The sizes of the corresponding $N^2$ stripe intervals are determined by the respective rates of the corresponding VOQs,
and their placements to consecutive intermediate ports are performed using a randomized algorithm.
Once generated, their placements remain fixed thereafter, while their sizes 
could change when their respective rates do.
Our goal in designing this randomized algorithm is that, 
when switched through the resulting (random) stripe intervals,  
traffic going out of any input port or going into any output port is near-perfectly balanced 
across all $N$ intermediate ports, with overwhelming probabilities.  

Once the $N^2$ stripe intervals are generated and fixed, packets in each VOQ will be striped across its
corresponding interval of intermediate ports as follows.  
Fix an arbitrary VOQ, and let its stripe size and interval be $2^k$ and $(\ell, \ell+2^k]
\equiv \{\ell+1, \ell+2, \cdots, \ell+2^k\}$, respectively.
(The integer $\ell$ must be divisible by $2^k$ for the interval to be dyadic, as discussed earlier.)
Packets in this VOQ are divided, chronologically according to their arrival times, 
into groups of $2^k$ packets each and,
with a slight abuse of the term, we refer to each such group also as a stripe.  Such a
stripe will eventually be switched through the set of intermediate ports $\{\ell+1, \ell+2, \cdots, \ell+2^k\}$
as follows. Assume this switching operation starts at time (slot) $t$, when the corresponding input port 
is connected to the intermediate port $\ell+1$ by the first switching fabric.  Then, following the periodic connection 
sequence of the first switching fabric, the input port forwards the first packet in the stripe to 
the intermediate port $\ell+1$ at time $t$, the second packet to the intermediate port $\ell+2$ at time $t + 1$, and so on,
until it forwards the last packet in the stripe to the intermediate port $\ell+2^k$ at time $t+ 2^k -1$.  That is, packets in this stripe
go out of the input port to consecutive intermediate ports in consecutive time slots (i.e., ``continuously").  The same can be said about other stripes from this and other VOQs.
This way, at each input port, the (switching) service is rendered by the first switching fabric in a stripe-by-stripe manner (i.e., finish serving 
one stripe before starting to serve another).


At each input port, packets from the $N$ VOQs originated from it compete for (the switching) service by the first switching fabric
and hence must be arbitrated by a scheduler.   However, since the service is rendered stripe-by-stripe, as explained above, the scheduling
is really performed among the competing stripes.  In general, two different scheduling policies are needed in a Sprinklers switch, one used
at the input ports and the other at intermediate ports.  
Designing stripe scheduling policies that are well suited for a Sprinklers switch 
turns out to be a difficult undertaking because two tricky design requirements have to be met simultaneously.  

The first requirement is that the resulting scheduling policies must facilitate a highly efficient utilization of the switching capacity of both switching fabrics.
To see this, consider input port $i$, whose input link has a normalized rate of 1, and intermediate port $\ell$.  
As explained in 
Sec.~\ref{sec:intro}, these two ports are connected only once every $N$ time slots, by the first switching fabric.  
Consider the set of packets at input port $i$ that needs to be switched to the intermediate port $\ell$.  
This set can be viewed and called a queue (a queueing-theoretic concept), because 
the stripe scheduling policy naturally induces a service order (and hence a service process)
on this set of packets, and the arrival time of each packet is simply that of the stripe containing the packet.
Clearly, the service rate of this queue is exactly $1/N$.  Suppose this input port is heavily loaded 
say with a normalized arrival rate of 0.95.  Then even with perfect load-balancing, the arrival rate to this queue is $0.95/N$, only slightly below the service rate $1/N$.  
Clearly, for this queue to be stable,
there is no room for any waste of this service capacity.
In other words, the scheduling policy must be throughput optimal.

The second requirement is exactly why we force all packets in a VOQ to go down the same fat path (i.e., stripe interval) 
through a Sprinklers switch, so as to guarantee that no packet reordering
can happen.  
It is 
easy to verify
that for packet order to be preserved within every stripe (and hence within every VOQ),
it suffices to guarantee that packets in every stripe go into their destination output port from
consecutive intermediate ports in consecutive time slots (i.e., continuously).  However, it is much harder for a stripe of $2^k$ packets to arrive at an output port continuously (than to leave an input port),
because these packets are physically located at a single input port when they leave the input port, but 
across $2^k$ different intermediate
ports right before they leave for the output port.  

After exploring the entire space of stripe-by-stripe scheduling policies,
we ended up adopting the same stripe scheduling policy, namely Largest Stripe First (LSF),
at both input ports and intermediate ports, but not for the same set of reasons.
At input ports, LSF is used because it is throughput optimal.
At intermediate ports,
LSF is also used because it seems to be the only policy that makes every stripe of packets arrive at their destination output port
continuously without incurring significant internal communication costs between the input and intermediate ports.
In Sec.~\ref{sec: lsf scheduling},
we will describe the LSF policies at the input and intermediate ports.



\subsection{Generating Stripe Intervals}
\label{sec:stripe generation}


As we have already explained, the random permutation and the variable dyadic stripe sizing techniques are used to
generate the stripe intervals for the $N$ VOQs originated from a single input port, with the objective of 
balancing the traffic coming out of this input port very evenly across all $N$ intermediate ports.
In Sec.~\ref{subsec: size determination}, 
we will specify this interval generating process precisely and provide detailed rationales for it.  
In Sec.~\ref{sec:stripe size eq}, we will 
explain why and how the interval generating processes at $N$ input ports should be carefully coordinated.

\subsubsection{Stripe interval generation at a single input port}
\label{subsec: size determination}

Suppose we fix an input port and number the $N$ VOQs originated from it as $1, 2, \cdots, N$, respectively.  
VOQ $i$ is first mapped to a distinct {\it primary intermediate port}
$\sigma(i)$, where $\sigma$ is a permutation chosen uniformly randomly from the set of all permutations on the set $\{1, 2, \cdots, N\}$.  
Then the stripe interval for a VOQ $i$ whose primary intermediate port is $\sigma(i)$ and
whose interval size is $n$, which must be a power of 2 as explained earlier, is simply the unique dyadic
interval of size $n$ that contains $\sigma(i)$.  A dyadic interval is one resulting from dividing the whole port number range 
$(0, N] \equiv \{1, 2, \cdots, N\}$ evenly by a power of 2, which takes the form
$(2^{k_0} m, 2^{k_0} (m + 1)]$ where
$k_0$ and $m$ are nonnegative integers.

\begin{figure}[htb]
  \centering
    \includegraphics[width=0.48\textwidth]{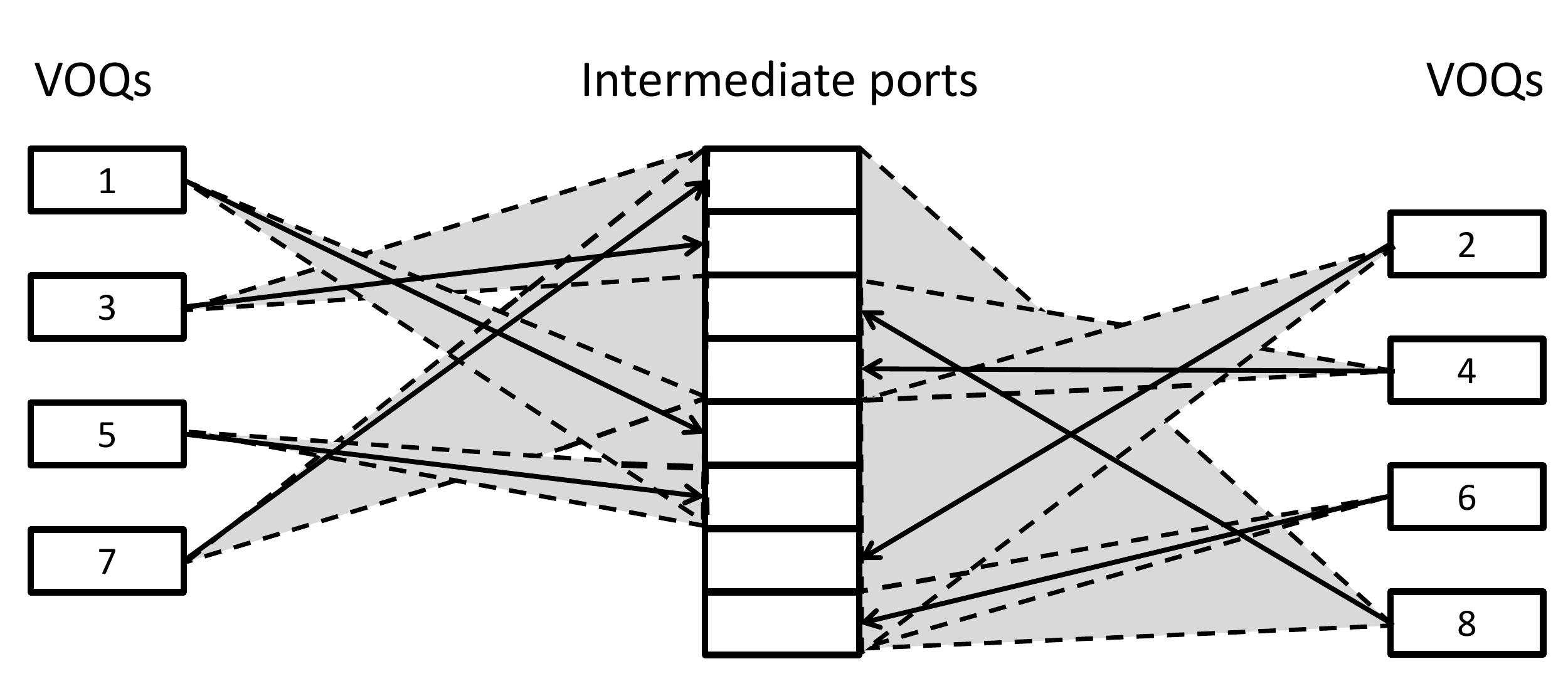}
    \caption{An example of VOQs being mapped to stripe intervals.}
\label{fig:stripe-interval-mapping}
\end{figure} 

In Fig.~\ref{fig:stripe-interval-mapping}, we show an example of mapping 8 VOQs at a single input port to their stripe intervals
in an 8 by 8 switch.  The middle column stands for the 8 intermediate ports and the columns
on the left and on the right correspond to VOQs 1, 3, 5, 7 and VOQs 2, 4, 6, 8, respectively.  Mapping a VOQ to 
its primary intermediate port is indicated by a solid arrow connecting the two.  Its stripe interval is indicated by 
the ``fanout" shadow that contains the solid line and is bounded at both ends by two dashed lines.  For example, 
VOQ 7 (bottom one on the left) is mapped to the primary intermediate 
port 1 and its stripe size is 4.  Therefore, its stripe interval is $(0, 4]$, the size-4 dyadic interval that contains intermediate
port 1.


\subsubsection{Stripe Size Determination}
\label{sec:stripe size eq}

As mentioned
in the ``water pressure equalization" analogy, for a VOQ, the size of its stripe interval is 
roughly proportional to the (current) traffic rate of the VOQ.  More precisely, the stripe interval size for
a VOQ of traffic rate $r$, is determined by the following function:
\begin{align}
  \label{eqn:stripe size rule}
  F(r) = \min\left\{N, 2^{\ceil{\log_2 (rN^2)}}
  \right\}.
\end{align}
The stripe size determination rule~(\ref{eqn:stripe size rule}) 
tries to bring the amount of traffic each intermediate port within the stripe interval receives from 
this VOQ below $1/N^2$ while requiring that the stripe size 
be a power of 2. 
However, if the rate $r$ is very high, say $>\frac{1}{N}$, then the stripe size $F(r)$ is simply $N$.

The initial sizing of the $N^2$ stripe intervals may be set based on historical switch-wide traffic matrix
information, or to some default values.  
Afterwards,
the size of a stripe interval will be adjusted based on the measured rate of the corresponding VOQ.  
To prevent the size of a stripe from ``thrashing" between $2^k$ and $2^{k+1}$, we can delay the 
halving and doubling of the stripe size.  Although our later mathematical derivations assume the perfect
adherence to the above stripe size determination rule, the change in the provable load-balancing guarantees
is in fact negligible when a small number of stripe sizes are a bit too small (or too large) for the respective rates of 
their corresponding VOQs.






\subsubsection{Coordination among all $N^2$ stripe intervals}
\label{sec:OLS}


Like an input port, an output port is connected to each intermediate port also exactly once every $N$
time slots, and hence needs to near-perfectly balance the traffic load coming into it.  That is, roughly $\frac{1}{N}$ of 
that load should come from each intermediate port.  In this section, we show how to achieve such load balancing at
an output port, by a careful coordination among the (stripe-interval-generating) permutations $\sigma_i$, $i = 1, 2, \cdots, N$.

Now consider an arbitrary output port $j$.  There are precisely $N$ VOQs destined for it, one originated from 
each input port.  Intuitively, they should ideally be mapped (permuted) to $N$ distinct primary intermediate ports 
-- for the same reason why the $N$ VOQs originated from an
input port are mapped (permuted) to $N$ distinct primary intermediate ports --
by the permutations $\sigma_1$, $\sigma_2$, $\cdots$, $\sigma_N$, respectively.  We will show this property 
holds for every output port $j$ if and only if the matrix representation of these permutations is an
Orthogonal Latin Square (OLS).

Consider an 
$N \times N$
matrix $A = (a_{ij})$, $i, j = 1, 2, \cdots, N$,  
where $a_{ij} = \sigma_i(j)$.
Consider the $N$ VOQs originated at input port $i$, one destined for each output port.  
We simply number 
each VOQ by its corresponding output port.
Clearly row $i$ of the matrix $A$ is the primary intermediate 
ports of these $N$ VOQs.   Now consider the above-mentioned $N$ VOQs destined for output port $j$, one originated from
each input port.   It is easy 
see that the $j_{th}$ column of matrix $A$, namely 
$\sigma_1(j)$, $\sigma_2(j)$, $\cdots$, $\sigma_N(j)$, is precisely the primary intermediate ports to which these
$N$ VOQs are mapped.  As explained earlier, we would like these numbers also to be distinct, i.e., a permutation of the set $\{1, 2, \cdots, N\}$.
Therefore, every row or column of the matrix $A$ must be a permutation of $\{1, 2, \cdots, N\}$.  Such a matrix is called an 
OLS in the combinatorics literature~\cite{colbourn1996crc}.


Our worst-case large deviation analysis in Sec.~\ref{sec:stability} requires the $N$ VOQs at the same input port or 
destined to the same output port
select their primary intermediate ports according to a uniform random permutation.
Mathematically, we only require the marginal distribution of the permutation represented by each row or column 
of the 
OLS
$A$ to be uniform.  
The use of the word ``marginal" here emphasizes that
we do not assume any dependence structure, or the lack thereof, among the $N$ random permutations 
represented by the $N$ rows and among those represented by the $N$ columns.  We refer to such an OLS as being weakly uniform random, to 
distinguish it from an OLS sampled uniformly randomly from the space of all
OLS' over the alphabet set $\{1, 2, \cdots, N\}$, which we refer to as being strongly uniform random.  
This distinction is extremely important for us, since  
whether there exists a polynomial time randomized algorithm for generating an OLS that is approximately strongly
uniform random has been an open problem in theoretical computer science and combinatorics 
for several decades~\cite{drizen2012generating,jacobson1996generating}.
A weakly uniform random OLS,
on the other hand, can be generated in $O(N \log N)$ time, shown as following.

We first generate two 
uniform random permutations $\sigma^{(R)}$
and $\sigma^{(C)}$ over the set $\{1, 2, \cdots, N\}$ that are mutually independent,
using a straightforward randomized algorithm~\cite{durstenfeld1964algorithm}.  
This process, which involves generating
$\log_2 N!$ = $O(N \log N)$ random bits needed to ``index" 
$\sigma^{(R)}$ and $\sigma^{(C)}$ each,  has $O(N \log N)$ complexity in total.  Then each matrix 
element $a(i, j)$ is simply set to $(\sigma^{(R)}(i) + \sigma^{(C)}(j) \mod N) + 1$.   It is not hard 
to verify that each row or column is a uniform random permutation.


\subsection{Largest Stripe First Policy}
\label{sec: lsf scheduling}



In this section, we describe our stripe scheduling policy called Largest Stripe First (LSF), which is used at both
input and intermediate ports of a Sprinklers switch.  LSF can be implemented in a straightforward
manner, at both input and intermediate ports, using $N (\log_2 N + 1)$ FIFO queues (a data structure concept).
Using an $N \times (\log_2 N + 1)$ 2D-bitmap to indicate the status of each queue (0 for empty, 1 for nonempty),
the switch can identify the rightmost bit set in each row of the bitmap in constant time,
which is used by LSF to identify the largest stripe to serve at each port.

We first provide a brief description of the periodic sequences of connections executed at both switching fabrics shown in Fig.~\ref{fig:load balance}.
The first switching fabric executes a periodic ``increasing'' sequence, that is, at any time slot $t$, each input port $i$ is connected
to the intermediate port $((i + t) \mod N) + 1$.  
The second switching fabric, on the other hand, will execute a periodic ``decreasing" sequence, that is, 
at any time slot $t$, each intermediate port $\ell$ is connected to the output port $((\ell - t) \mod N)+1$.  

In the rest of the paper, we make the following standard homogeneity assumption about a Sprinklers switch.  
Every input, intermediate, or output port
operates at the same speed.  That is, each can process and transmit exactly one packet per time slot.  We refer to this speed
as $1$.  Every connection made in a switching fabric also has speed of $1$ (i.e., one packet can be switched per time slot).
Since $N$ connections are made by a switching fabric at any time slot, up to $N$ packets can be switched during the time 
slot.





\subsubsection{Stripe Scheduling at Input Ports}
\label{sec:scheduling input}


\begin{algorithm}
\caption{LSF policy on ``Who is next?"}
\label{alg:select next VOQ to serve}
  \begin{algorithmic}[1]
  \State $l = (i + t \mod N) + 1$; 
    \If {No stripe is being served at time $t$}
      \State Let $S$ be the set of stripes with interval $(l-1, *]$;
      \If {$S$ is not empty}
        \State Start serving the largest stripe in $S$;
      \EndIf
    \EndIf
  \end{algorithmic}
\end{algorithm}


The above pseudocode describes the Largest stripe first (LSF) policy used at input port $i$ to make
a decision as to, among multiple competing stripes, which one to serve next.   
Note that by the stripe-by-stripe nature of the scheduler, it is asked to make such a policy decision only 
when it completely finishes serving a stripe.
The policy is simply to pick among the set of stripes
whose dyadic
interval starts at intermediate port $\ell$ -- provided that the set is nonempty --  the largest one (FCFS for tie-breaking) to start serving immediately.  
The LSF policy is clearly throughput optimal because it is work-conserving in the sense whenever a connection is made between input port $i$ and intermediate port $\ell$, a packet will be served
if there is at least one stripe at the input port 
whose interval contains $\ell$.  
Using terms from queueing theory~\cite{kleinrock1975queueing}
we say this queue is served by a work-conserving server with service rate $1/N$.





A stripe service schedule, which is the outcome of a scheduling policy acting on the stripe arrival process, 
is represented by a schedule grid, as shown in Fig.~\ref{fig: instance}.
There are $N$ rows in the grid, each corresponding
to an intermediate port.  Each tiny square in the column represents a time slot. 
The shade in the square represents a packet scheduled to be transmitted in that slot and
different shade patterns mean different VOQs.  Hence each ``thin vertical bar" of squares with the same shade 
represents a stripe in the schedule.
The time is progressing
on the grid from right to left in the sense that packets put inside the rightmost column 
will be served in the first cycle ($N$ time slots), the next column to the left in the second
cycle, and so on.  Within each column, time is progressing from up to down in the sense
that packet put inside the uppermost cell will be switched to intermediate port 1 say at
time slot $t_0$, packet put inside the cell immediately below it will be switched to intermediate
port 2 at time slot $t_0+1$, and so on.  Note this up-down sequence is consistent with the above-mentioned
periodic connection pattern between the input port and the intermediate ports.
Therefore, the scheduler is ``peeling" the grid from right to left and from up to down.

For example, Fig.~\ref{fig: instance}
corresponds to the statuses of the scheduling grid at time $t_0$ and $t_0 + 8$, 
where LSF is the scheduling policy.  We can see that
the rightmost
column 
in the left part of Fig.~\ref{fig: instance}, representing a single stripe of size $8$, 
is served by the first switching fabric in the meantime, 
and hence disappears in the right part of the figure.  
The above-mentioned working conservation nature of the LSF server for any input port is
also clearly illustrated in Fig.~\ref{fig: instance}:  
There is no ``hole" in any row.  

\begin{figure}[htb]
  \centering
    \includegraphics[width=0.48\textwidth]{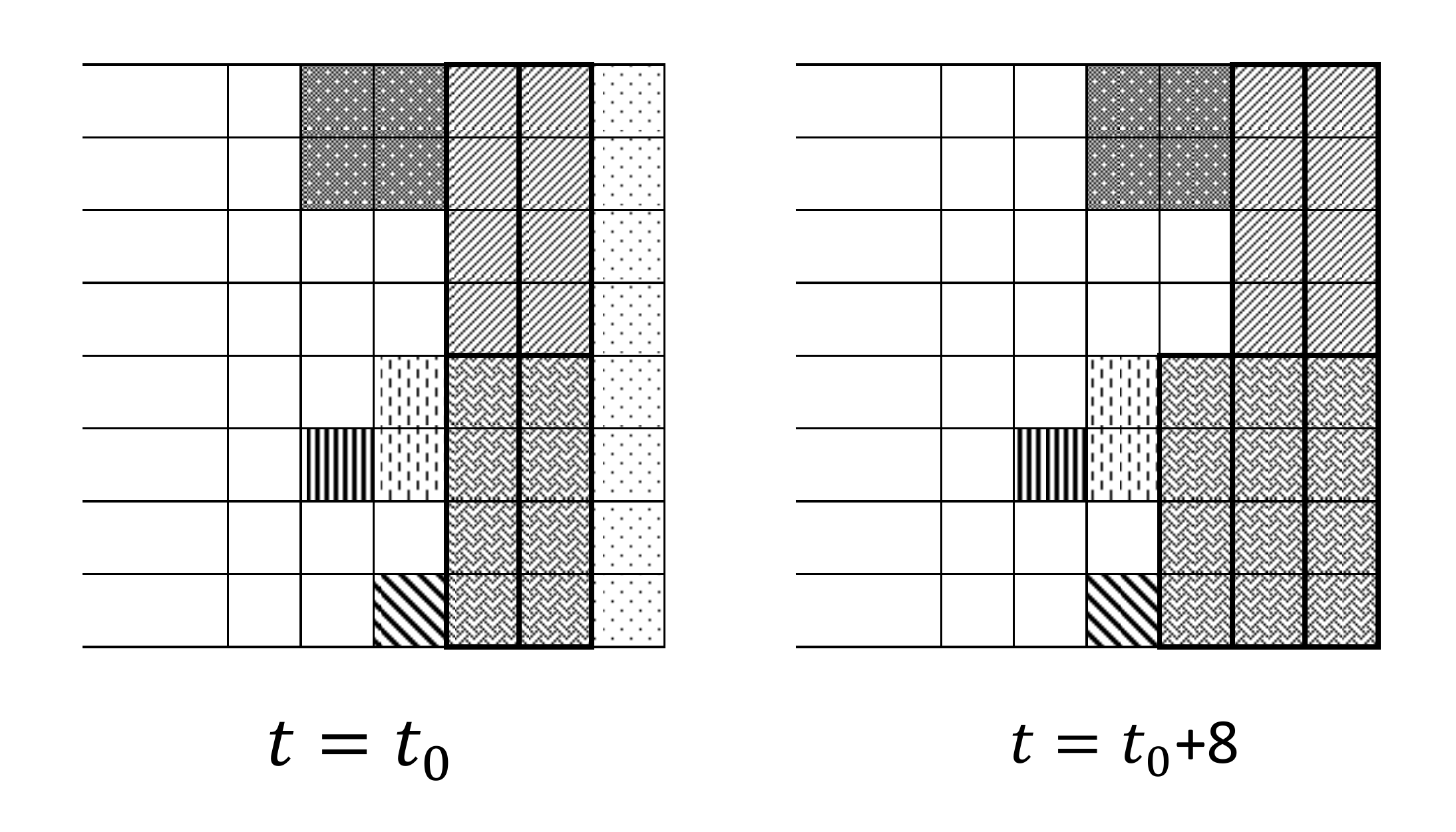}
    \caption{Statuses of the stripe scheduling at time $t_0$ (left) and $t_0+8$ (right) respectively.}
\label{fig: instance}
\end{figure}	 






While the scheduler always peels the grid from right to left and from up to down, 
the scheduler may insert a stripe ahead of other stripes smaller than it in size.  
In this case, the planned service time of all cells and bars to the left of the inserted vertical bar 
are shifted to the left by 1 column, meaning that their planned service time will 
be delayed by $N$ time slots.
For example, we 
can see from comparing the two parts of Fig.~\ref{fig: instance}
that a stripe of size 4 is inserted in the lower part of the 
3rd column in the right part of Fig.~\ref{fig: instance}, 
after two other stripes of size 4, but before some other stripes of smaller sizes.





\subsubsection{Implementation of LSF at an Input Port}
\label{subsec:data strucutre lsf}

In this section, we describe data structures and algorithms for implementing the Longest 
Stripe First scheduling policy at an input port.  Again we use input port 1 as an illustrative example.
The LSF policy at input port 1 can be implemented using $N (\log_2 N+1)$ 
FIFO queues as shown in Fig.~\ref{fig:data structure}.  Conceptually, these FIFO queues are arranged into an array with 
$N$ rows and $\log_2 N + 1$ columns.  Row $\ell$ corresponds to FIFO queues that buffer packets 
to be sent to intermediate port $\ell$.  Each column of FIFO queues is used to buffer stripes of a certain size.  
The last column is for stripes of size $N$, the second last for stripes of size $N/2$, and so on.  
The very first column is for stripes of size $1$.  Therefore, the FIFO on the $\ell$th row and $k$th column
is to queue packets going to intermediate port $\ell$ and from stripes of size
$2^{k-1}$.

\begin{figure}[htb]
  \centering
    \includegraphics[width=0.48\textwidth]{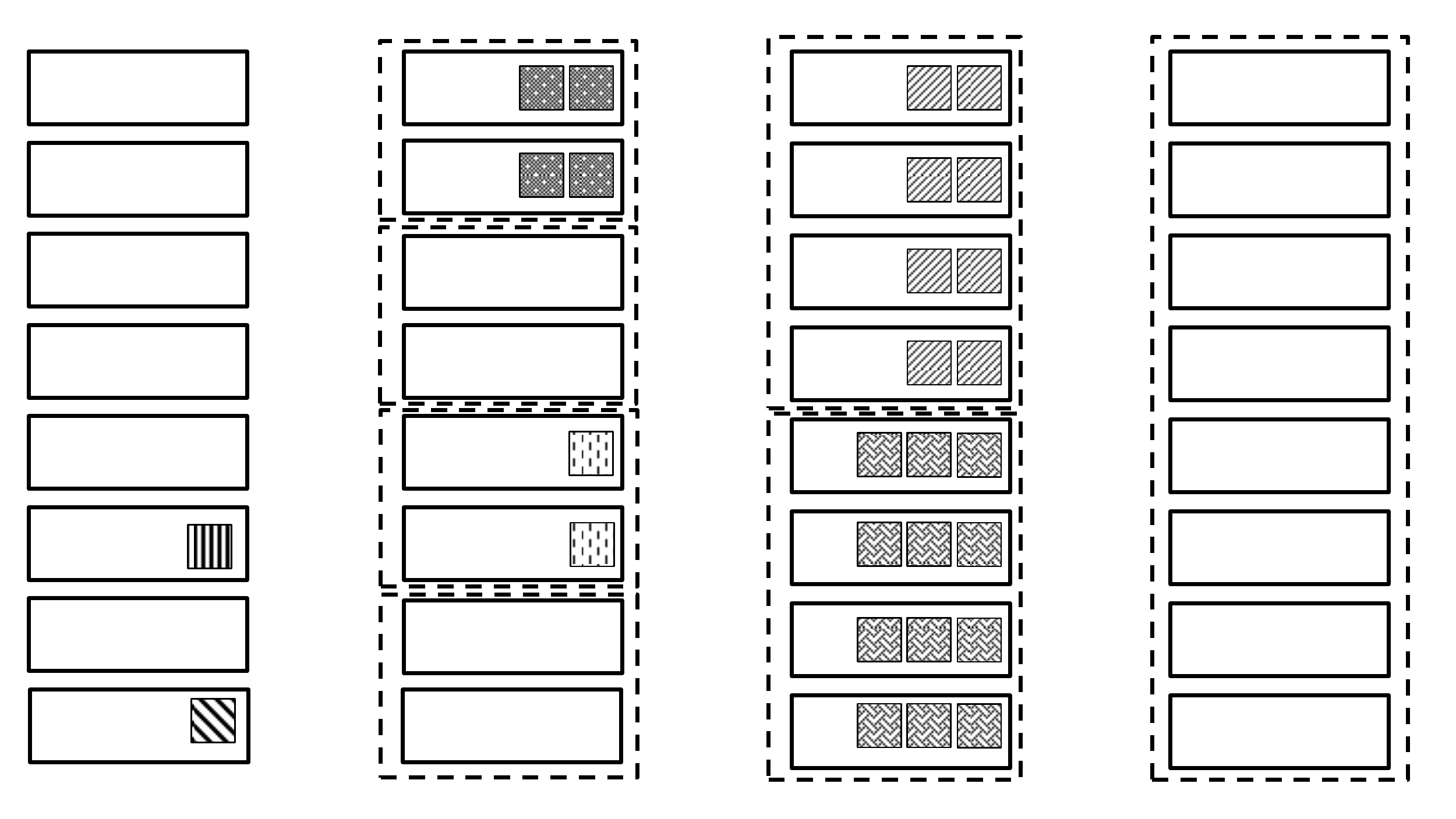}
    \caption{Data structure to implement LSF.}
\label{fig:data structure}
\end{figure}


The implementation of LSF with this data structure is straightforward.  Each VOQ will have a ready queue 
to buffer packets that have not filled up a stripe yet.  Once a stripe is full, depending on the size of the stripe, 
it will be ``plastered" into the corresponding interval in the corresponding column.  At each column,
stripes sent into the same interval will be queued in the FIFO order.  For example, the contents of 
these FIFO queues at time $t_0+8$, that correspond to the scheduling grid shown as the right part of Fig.~\ref{fig: instance}, 
are shown in Fig.~\ref{fig:data structure}.

During each time slot $t$, the scheduler either starts or is in the middle of serving the longest stripe that contain the intermediate
port $t \mod N$.  However, the scheduler 
does not need to be stripe-aware.  At each time $t$, it only needs to scan the $(t \mod N)th$ row 
from right to left and serve the packet at the head of the first nonempty queue.  If for each queue, we use a bit to encode whether it
is empty, our problem boils down to looking for the ``first one from right" in a bitmap $\log_2 N + 1$ bits long.  This kind of operation
can be done in one to a few clock cycles on modern processors or ASIC chip.  Overall $N$ such bitmaps are needed, one for each 
row (intermediate port).

Careful readers may note that it is unnecessary to maintain $N \log_2 N$ FIFO queues.  Rather, we need only one FIFO queue for all 
stripes of size $N$, two for all stripes of size $N/2$, -- namely one for the interval (0, N/2] and the other for $(N/2, N]$), -- four for all stripes of size
$N/4$, and so on.  The total number of FIFO queues we need is indeed $1 + 2 + 4 + \cdots + N/2 + N = 2N-1$.  
For example, in the above figure, we need only 15 queues: one FIFO queue for all stripes of size 8, two for all stripes of size 4, 
four for all stripes of size 2, and 8 for all stripes of size 1.  We will indeed adopt such a simplified implementation for the 
input port.  However, we still present the data structure in this ``verbose" manner because this data structure will be modified for use at
the intermediate ports, and there such a simplifying collapse is no longer possible, as will become clear next.


\subsubsection{Stripe Scheduling at Intermediate Ports}
\label{sec:scheduling intermediate}

In Sprinklers, an intermediate port also adopts the LSF policy for two reasons.  First, throughput-optimality is needed also at an intermediate port for the same reason stated above.  Second, given the distributed nature
in which any scheduling policy at an intermediate port is implemented, LSF 
appears to be the easiest and least costly to implement.

The schedule grids at intermediate ports take a slightly different form.
Every output port $j$ is associated with a schedule grid also consisting of $N$ rows. 
Row $i$ of the grid corresponds to the tentative schedule in which packets destined for output port 
$j$ at intermediate port $i$ will follow.  This schedule grid is really a virtual one, of which the $N$ rows are physically distributed across $N$ different intermediate ports respectively.  All stripes heading to output 
$j$ show up on this grid.  The LSF policy can be defined with respect to this virtual grid in 
almost the same way as in the previous case.  

The data structure for implementing LSF at intermediate ports is the same as that at input ports, except 
components of each instance are distributed across all $N$ intermediate ports, thus requiring some 
coordination.  This coordination however requires only that, for each packet switched over the first switching fabric, 
the input port inform the intermediate port of the size of the stripe to which the packet belongs.  This 
information can be encoded in just $\log_2 \log_2 N$ bits, which is a tiny number (e.g., $=4$ bits 
when $N = 4096$), and be included in the internal-use header
of every packet transmitted across the first switching fabric.

\section{Stability analysis}
\label{sec:stability}

As explained in Sec.~\ref{sec: lsf scheduling}, a queueing process can be defined on the set of packets
at input port $i$ that need to be switched to intermediate port $\ell$.
This (single) queue is served by a work-conserving server with a service rate $1/N$ under the LSF scheduling policy.
It has been shown
in \cite{dai2000throughput} that such a queue is stable
as long as the long-term average arrival rate is less than $1/N$.
This entire section is devoted to proving a single mathematical result, that is, the arrival rate to
this queue is less than $1/N$ with overwhelming 
probability, under all admissible traffic workloads to the input port $i$.

Between input ports $i = 1, 2, ..., N$ and intermediate ports $\ell = 1, 2, ..., N$, 
there are altogether $N^2$ such queues to analyze.  It is not hard  to verify, however, that 
the joint probability distributions of the random variables involved in -- and hence the results from --
their analyses are identical.
We will present only the case for $i= 1$ and $\ell = 1$.  
Thus, the queue in the rest of the section specifically refers to the one serving packets from input port
1 to intermediate port 1.
With this understanding, we will drop subscripts
$i$ and $\ell$ in the sequel for notational convenience.

There are another $N^2$ queues to analyze, namely, the queues of packets that need to be
transmitted from intermediate ports $\ell = 1, 2, ..., N$ to output ports $j = 1, 2, ..., N$.  
Note each such queue is also served
by a work-conserving
server with service rate $1/N$, because the LSF scheduling policy is also used at the intermediate ports, 
as explained in Sec.~\ref{sec:scheduling intermediate}.
Therefore, we again need only to prove that the arrival rate to each such queue is statistically less than 
$1/N$.  However, again 
the joint probability distributions of the random variables involved in -- and hence the results from --
their analyses are identical to that of the $N^2$ queues between the input ports and the intermediate ports,
this time due to the statistical rotational symmetry of the above-mentioned weakly random OLS based on 
which the stripe intervals are generated.  Therefore, no separate analysis is necessary for these 
$N^2$ queues.




\subsection{Formal statement of the result}


In this section, we precisely state the mathematical result. 
Let the $N$ VOQs at input port $1$ be numbered $1$, $2$, $\cdots$, $N$.  Let their arrival rates 
be $r_1$, $r_2$, $\cdots$, and $r_N$, respectively, which we write also in the vector form $\vec{r}$.
Let $f_i$ be the stripe size of VOQ $i$, i.e., $f_i := F(r_i)$ for $i = 1, 2, \cdots, N$, where $F(\cdot)$ 
is defined in Equation~\ref{eqn:stripe size rule}.  
Let $s_i$ be the load-per-share of VOQ $i$, i.e., $s_i := \frac{r_i}{F(r_i)}$ for $i = 1, 2, \cdots, N$.
Let $|\vec{r}|$ be the sum of the arrival rates, i.e., $|\vec{r}| := r_1 + \cdots + r_N$.

Let $X(\vec{r}, \sigma)$ be the total traffic arrival rate to the queue when the $N$ input VOQs have
arrival rates given by $\vec{r}$ and the uniform random permutation used to map the $N$ VOQs to their
respective primary intermediate ports is $\sigma$.  All $N$ VOQs could 
contribute some traffic to this queue (of packets that need to be switched to
intermediate port $1$) and $X(\vec{r}, \sigma)$ is simply the sum of all these contributions.
Clearly $X(\vec{r}, \sigma)$ is a random variable whose value depends on the VOQ rates $\vec{r}$ 
and how $\sigma$ shuffles these VOQs around.  

Consider the VOQ that selects intermediate port $\ell$ as its primary intermediate port.
By the definition of $\sigma$, the index of this VOQ is $\sigma^{-1}(\ell)$.
It is not hard to verify that this VOQ contributes one of its load shares to this queue,
in the amount of $s_{\sigma^{-1}(\ell)}$, if and only if its stripe size is at least $\ell$.
Denote this contribution as $X_\ell(\vec{r}, \sigma)$.  we have
$X_\ell(\vec{r}, \sigma) = s_{\sigma^{-1}(\ell)} 
\mathbbm{1}_{\{f_{\sigma^{-1}(\ell)} \ge \ell\}}$.  Then:
\begin{align}
	\label{eqn:decomposition of Y using x}
	X(\vec{r}, \sigma) = \sum_{\ell =1 }^N X_\ell(\vec{r}, \sigma).
\end{align}

Note all the randomness of $X_\ell(\vec{r}, \sigma)$ and hence $X(\vec{r}, \sigma)$ comes
from $\sigma$, as $\vec{r}$ is a set of constant parameters.  The expectations of the functions of $X_\ell(\vec{r}, \sigma)$ and $X(\vec{r}, \sigma)$, such as their moment generating functions (MGF), are taken over $\sigma$.  With this understanding,
we will drop $\sigma$ from $X_\ell(\vec{r}, \sigma)$ and $X(\vec{r}, \sigma)$ and simply write them as $X_\ell(\vec{r})$ and $X(\vec{r})$ in the sequel.


As explained earlier, the sole objective of this section is to prove that $\mathbb{P}(X(\vec{r}) \ge 1/N)$, the 
probability that the total arrival rate to this queue exceeds the service rate, is either 0 or extremely small.  
This probability, however, generally depends not only on
$|\vec{r}|$, the total traffic load on input port $1$, but also on how this total splits
up (i.e., the actual rate vector $\vec{r}$).  

Our result is composed of two theorems, namely,
Theorems 1 and 2.  Theorem 1 states that when this total load is no more than roughly $2/3$, then 
this probability is strictly $0$, regardless of how it is split up.  
When the total traffic load $|\vec{r}|$ exceeds that amount, however, this 
probability could become positive, and how large this probability is could hinge heavily on 
the actual $\vec{r}$ values.  Theorem 2, together with the standard Chernoff bounding technique 
preceding it, provides a bound on the maximum value of 
this probability when the total load is no more than some constant $\rho$, but can be split up in arbitrary ways.
More precisely, Theorem 2 helps  
establish an upper bound on 
$\sup_{\vec{r}\in V(\rho)} \mathbb{P}(X(\vec{r}) \ge 1/N)$, where
$V(\rho) := \{ \vec{r} \in \mathbb{R}^N_+: |\vec{r}| \le \rho \}$.
However, since $\mathbb{P}(X(\vec{r}) \ge 1/N)$, as a function of $\vec{r}$, is nondecreasing, this
is equivalent to bounding $\sup_{\vec{r}\in U(\rho)} \mathbb{P}(X(\vec{r}) \ge 1/N)$, where
$U(\rho) := \{ \vec{r} \in \mathbb{R}^N_+: |\vec{r}| = \rho \}$.  The same 
argument applies to $\mathbb{E}[\exp(\theta X(\vec{r}))]$, the MGF of
$X(\vec{r})$.  Therefore, we will use $U(\rho)$ instead of $V(\rho)$ in both cases in the sequel.

\begin{theo}
	\label{theo:threshold of rho}
	$X(\vec{r}) < 1/N$ with probability $1$ if $|\vec{r}| < \frac{2}{3} + \frac{1}{3N^2}$.
\end{theo}

Let $r_\ell$ be the arrival rate of the VOQ that selects intermediate port $\ell$
as primary port. The corresponding stripe size and load-per-share of $r_\ell$ is given by
$f_\ell$ and $s_\ell$. Given the vector $\vec{r} = (r_1, \cdots, r_N)$, 
the arrival rate the queue 
(at input port 1, serving packets going to intermediate port 1) receives
is $X(\vec{r}) = \sum_{\ell = 1}^N s_\ell \mathbbm{1}_{\{f_\ell \ge \ell\}}$. We prove
Theorem \ref{theo:threshold of rho} by studying the following problem.
\begin{align}
	\label{prob:threshod rho}
	\min_{r}:& \sum_{\ell = 1}^N r_\ell \\
	\nonumber
	s.t. & \sum_{\ell = 1}^N s_\ell \mathbbm{1}_{\{f_\ell \ge \ell\}} \ge 1/N
\end{align}
\begin{lemm}
\label{lemm:fill water}
	Suppose $r^* = (r_1^*, \cdots, r_N^*)$ is an optimal solution to Problem 
	(\ref{prob:threshod rho}). Then $f^*_\ell = 2^{\ceil{\log_2 \ell}}$ if
	$r^*_\ell > 0$, for $\ell = 1, \cdots, N$.
\end{lemm}
\begin{proof}
	We prove the lemma by contradiction. 
	\begin{enumerate}
		\item Suppose there is $r^*_{\ell_0} > 0$ and 
		$f_{\ell_0} < 2^{\ceil{\log_2 \ell_0}}$, 
		then $r^*_{\ell_0}$ does not contribute any arrival rate to the queue. Thus, we
		can reduce $r^*_{\ell_0}$ to zero without changing the 
		total arrival rate to the queue. This contradicting the assumption that $r^*$ is an
		optimal solution to Problem (\ref{prob:threshod rho}).
		\item Suppose there is $r^*_{\ell_0} > 0$ 
		and $f_{\ell_0}^* > 2^{\ceil{\log_2 \ell_0}}$.
		According to stripe size rule (\ref{eqn:stripe size rule}), we know 
		$s^*_{\ell_0} > 1/(2N^2)$.
		\begin{enumerate}
			\item If $s^*_{\ell_0} \le 1/N^2$, we construct solution $r' = r^*$ except 
			we set $r'_{\ell_0} = 2^{\ceil{\log_2 \ell_0}} / N^2$. 
			$f^*_{\ell_0} > 2^{\ceil{\log_2 \ell_0}}$ implies 
			$f^*_{\ell_0}) \ge 2 \times 2^{\ceil{\log_2 \ell_0}}$.
			Further since $s^*_{\ell_0} > 1/(2N^2)$, we have 
			that $r'_{\ell_0} < r^*_{\ell_0}$. It is easy to see that
			$X(r') > X(r^*)$, thus $r'$ is also a feasible solution and has smaller 
			total arrival rate than $r^*$. This contradicts the assumption that $r^*$ 
			is an optimal solution.
			\item If $s^*_{\ell_0} > 1/N^2$, implying $f^*_{\ell_0} = N$, we suppose 
			$r^*_{\ell_0} = 1/N + N \varepsilon$. 
			Since $N = f^*_{\ell_0} > 2^{\ceil{\log_2 \ell_0}}$, we have $\ell_0 \le N/2$.
			Construct $r' = r^*$ except set 
			$r'_{\ell_0} = 2^{\ceil{\log_2 \ell_0}} / 2N^2 + \delta$ and
			$r'_{N/2 + 1} = r^*{N/2+1} + (1/(2N^2) + \varepsilon) N$ where 
			$0 < \delta \ll \varepsilon$. Then 
			$X(r') = X(r^*) + \delta / 2^{\ceil{\log_2 \ell_0}} > X(r^*)$ and 
			the total arrival rate of $r'$ is smaller than that of $r^*$. This contradicts the 
			assumption that $r^*$ is an optimal solution.
		\end{enumerate}
	\end{enumerate}
	To conclude, we establish the property in the lemma.
\end{proof}
With the stripe size property in Lemma \ref{lemm:fill water}, if $r^*$ is an optimal
solution to Problem (\ref{prob:threshod rho}), then
all $r^*_1$ goes to the queue, 1/2 of $r_2^*$ goes to the queue, 1/4 of $r_3^*$ and $r_4^*$ goes to 
the queue, and so on. To minimize the total arrival rate while have $X(r^*) \ge 1/N$, one
should fill up $r_\ell^*$ that is closer to the queue first. Thus, 
$r^*_1 = 1/N^2, r^*_2 = 2/N^2, \cdots, r^*_{N/4+1} = \cdots = r^*_{N/2} = N/2 \cdot 1/N^2$.
From these VOQs, the queue receives total arrival rate $N/2 \cdot 1/N^2 = 1/(2N)$.
To achieve $X(r^*) = 1/N$, we set $r^*_{N/2+1} = N \cdot 1/(2N) = 1/2$ 
and all other $r_\ell^* = 0$. This optimal $r^*$ results in total arrival rate 
$(1 + 2 + 4 + \cdots + N/4 \cdot N/2) / N^2 + 1/2 = 2/3 + 1/(3N^2)$.






We now move on to describe the second theorem.
By the standard Chernoff bounding technique, we have
\begin{align*}
	\mathbb{P}(X(\vec{r}) \ge 1/N) \le \inf_{\theta > 0}
	\exp(-\theta /N)\mathbb{E} 
	[\exp(\theta X(\vec{r}))].
\end{align*}


The above-mentioned worst-case probability can thus be upper-bounded as
\begin{align}
	\nonumber
	&\sup_{\vec{r} \in U(\rho)} \mathbb{P}(X(\vec{r}) \ge 1/N) \\
	\le& \nonumber
	\sup_{\vec{r} \in U(\rho)} \inf_{\theta > 0} \exp(-\theta/N)
	\mathbb{E}[\exp(\theta X(\vec{r}))]\\
	\le& 
	\inf_{\theta > 0} \sup_{\vec{r}\in U(\rho)} \exp(-\theta/N)
	\mathbb{E}[\exp(\theta X(\vec{r}))],
	\label{eqn:math problem}
\end{align}
where the interchange of the infimum and supremum 
follows from the max-min inequality~\cite{boyd2004convex}.

Therefore, our problem boils down to upper-bounding $\sup_{\vec{r}\in U(\rho)} \mathbb{E}[\exp(\theta X(\vec{r}))]$, 
the worst-case MGF
of $X({\vec{r}})$, which is established in the following theorem.



\begin{theo}
	\label{theo:overflow bound mgf}
When $\frac{2}{3} + \frac{1}{3N^2} \le \rho < 1$, we have
	\begin{align*}
		\sup_{\vec{r}\in U(\rho)} \mathbb{E}[\exp(\theta X(\vec{r}))] \le 
		(h(p^*(\theta \alpha), \theta\alpha))^{N/2}
		\exp(\theta \rho / N),
	\end{align*}
	where 
	$h(p,a) = p \exp\left(a (1-p)\right) + 
	(1-p) \exp\left(-a p\right)$, 
	and 
	$$p^*(a) = \frac{\exp(a) - 1 - a}{\exp(a)a - a}$$ 
	is the maximizer of $h(\cdot, a)$ for a given $a$.
\end{theo}

For any given switch size $N$ and a total load $\rho < 1$,
Theorem \ref{theo:overflow bound mgf} and the Chernoff bound 
(\ref{eqn:math problem}) establish an upper bound on the above-mentioned worst-case probability.
Table \ref{tbl:example of bounds} provides the values of this probability 
under various settings of $N$ and $\rho$.  For example, when a Sprinklers switch has 2048 input and output 
ports (i.e., $N = 2048$) and is $93\%$ loaded on input port $1$, this upper bound is only
$3.09 \times 10^{-18}$.  Note however that it upper-bounds only the probability that a single queue, 
namely the set of packets at input port 1 that need to be switched to the intermediate ports 1, is overloaded.
As mentioned earlier, there are altogether $2N^2$ such queues.  Per the union bound, the probability that
one or more such queues are overloaded across the entire $2048\times 2048$ switch 
is simply $2N^2$ times that.  In this example, the bound on the switch-wide overloading probability is
$1.30 \times 10^{-11}$.  

\begin{table}[!htbp]
	\caption{Examples of overload probability bound}
	\label{tbl:example of bounds}
	\centering
	\begin{tabular}{cccc}
		\hline
		$\rho$ & $N=1024$ & $N=2048$ & N = 4096\\\hline
		0.90 & $1.21\times 10^{-18}$ & $1.14 \times 10^{-29}$ & $6.10 \times 10^{-30}$\\
		0.91 & $3.06\times 10^{-15}$ & $4.91 \times 10^{-29}$ & $7.10 \times 10^{-30}$\\
		0.92 & $3.54\times 10^{-12}$ & $1.26 \times 10^{-23}$ & $9.10 \times 10^{-30}$\\
		0.93 & $1.76\times 10^{-9}$ & $3.09 \times 10^{-18}$  & $1.58 \times 10^{-29}$\\
		0.94 & $3.76\times 10^{-7}$ & $1.42 \times 10^{-13}$  & $2.00 \times 10^{-26}$\\
		0.95 & $3.50 \times 10^{-5}$ & $1.22 \times 10^{-9}$  & $1.48 \times 10^{-18}$\\
		0.96 & $1.41 \times 10^{-3}$ & $1.99 \times 10^{-6}$  & $3.97 \times 10^{-12}$\\
		0.97 & $2.50 \times 10^{-2}$ & $6.24 \times 10^{-4}$  & $3.90 \times 10^{-7}$\\
		\hline
	\end{tabular}
\end{table}

From these bounds, we can see that despite 
the use of the randomized and variable-size striping needed to prevent packet reordering and 
keeping the packet buffering delay within a reasonable range, a Sprinklers switch can achieve very high 
throughputs ($> 90\%$).  In fact, the actual overloading probabilities could be orders of magnitude smaller,
due to unavoidable relaxations used to obtain such bounds.
Another interesting phenomena reflected in
Table \ref{tbl:example of bounds} is that, given a traffic load, this bound decreases rapidly
as the switch size $N$ increases.  In other words, the larger a Sprinklers switch is, the higher the throughput
guarantees it can provide.
This is certainly a desired property for a scalable switch.

\subsection{Proof of Theorem \ref{theo:overflow bound mgf}}
\label{sec:main proof}

Recall $f_i = F(r_i)$ and $s_i = \frac{r_i}{f_i}$ for $i = 1, 2, \cdots, N$.
For the convenience of our 
later derivations, we rewrite $X$ as a function of $\vec{f}$ and $\vec{s}$
by defining $X(\vec{s}, \vec{f}) = X(\vec{s} \circ \vec{f})$, where $\vec{s} \circ \vec{f} = (f_1 \cdot s_1, 
f_2 \cdot s_2, \cdots, f_N \cdot s_N)$.  With this rewriting,
we can convert the optimization problem $\sup_{\vec{r} \in U(\rho)} 
\mathbb{E} [\exp(\theta X(\vec{r}))]$ in Theorem 2 to the following:
\begin{align}
\label{eqn:problem s F}
	\sup_{\langle \vec{f}, \vec{s} \rangle \in A(\rho)} 
	\mathbb{E} [\exp(\theta X(\vec{s}, \vec{f}))],
\end{align}
where 
\begin{align}
	A(\rho) := \{ 
		\langle \vec{f}, \vec{s} \rangle 
				| & f_\ell = 2^{k_\ell} \text{ for some integer } k_{\ell}, \forall \ell
				\label{eqn:power of 2 stripe size} 
			     \\
			     \label{eqn:per share condition 1}
			     & s_{\ell} \in [0, \alpha] \text{ if } f_{\ell} = 1,\\
			     \label{eqn:per share condition 2}
			     & s_{\ell} \in (\alpha/2, \alpha] \text{ if } f_{\ell} = 2, 4, \cdots, N/2\\
			     \label{eqn:per share condition N}
			     & s_{\ell} \in (\alpha/2, 1/N] \text{ if } f_{\ell} = N, \\
			     & \sum_{\ell = 1}^N s_\ell f_{\ell} = \rho
	\}.
\end{align}
The rule for stripe size determination~(\ref{eqn:stripe size rule}) induces the
conditions (\ref{eqn:power of 2 stripe size}) - (\ref{eqn:per share condition N}) above.


We now replace the half-open intervals in conditions (\ref{eqn:per share condition 2}) and (\ref{eqn:per share condition N}) by the closed ones, which is equivalent to replacing $A(\rho)$ with its closure $\bar{A}(\rho)$, because it is 
generally much easier to work with a closed set in an optimization problem.  
The resulting optimization problem is
\begin{align}
\label{eqn:problem s F closed}
	\max_{\langle \vec{f}, \vec{s} \rangle \in \bar{A}(\rho)} 
	\mathbb{E}_{\sigma} [\exp(\theta X(\vec{s}, \vec{f}))].
\end{align}
This replacement is logically correct because it will not decrease the optimal value of this optimization 
problem, and we are deriving an upper bound of it.

Our approach to Problem (\ref{eqn:problem s F closed}) is an exciting combination of convex optimization 
and the theory of negative associations in statistics~\cite{joag1983negative}, glued together by the 
Cauchy-Schwartz inequality.
We first show
that $\mathbb{E}[ \exp(\theta X(\vec{s}, \vec{f}))]$, viewed as a function of only $\vec{s}$ (i.e., 
with a fixed $\vec{f}$), achieves its maximum
over $\bar{A}(\rho)$ when $\vec{s}$ satisfies certain extremal conditions, through the convex optimization theory
in Sec.~\ref{sec:Extremal Property}.
Then with $\vec{s}$ fixed at any such extreme point, we can decompose $X(\vec{s}, \vec{f})$
into the sum of two sets of negatively associated (defined below) random variables.  Using properties of 
negative association, we are able to derive a very tight upper bound on $\mathbb{E}[ \exp(\theta X(\vec{s}, \vec{f}))]$ in
Sec.~\ref{sec:Negative Association}.




 \begin{defi}
 	Random variables $X_1$, $X_2$, $\cdots$, $X_k$ are said to be negatively associated
 	if for every pair of disjoint subsets $A_1$, $A_2$ of $\{1, 2, \cdots, k\}$,
 	\begin{align*}
 		Cov\{g_1(X_i,i \in A_1), g_2(X_j, j \in A_2)\} \le 0,
 	\end{align*}
 	whenever $g_1$ and $g_2$ are nondecreasing functions.
 \end{defi}


In the proof below, we need to use the following properties of negative association~\cite{joag1983negative}.
 \begin{lemm}
 \label{prty:na to product form}
 	Let 
 	$g_1, g_2, \cdots, g_k$ be nondecreasing positive functions of one variable. Then $X_1, \cdots,
 	X_k$ being negatively associated implies 
 	\begin{align*}
 		\mathbb{E} \left[ \prod_{i=1}^k g_i(X_i) \right] \le 
 		\prod_{i=1}^k \mathbb{E} \left[g_i(X_i) \right].
 	\end{align*}
 \end{lemm}


\begin{lemm}
\label{lemma:permutation is NA}
	Let $x = (x_1, \cdots, x_k)$ be a set of $k$ real numbers. $X=(X_1, \cdots, X_k)$ 
	is a random vector, taking
	values in the set of $k!$ permutations of $x$ with equal probabilities. Then
	$X=(X_1, \cdots, X_k)$ is negatively associated.
\end{lemm}

\subsubsection{Extremal Property of $\vec{s}$ in An Optimal Solution}
\label{sec:Extremal Property}

\begin{lemm}
\label{lemma:extreme value property}
	Given $\frac{2}{3} + \frac{1}{3N^2} \le \rho < 1$, for any $\theta > 0$, 
	there is always an optimal solution $\langle \vec{f}^*, \vec{s}^* \rangle$  to Problem
	(\ref{eqn:problem s F closed}) that satisfies the following property:	
	\begin{align}
		\label{eqn:extreme property}
		s_j^* = \left \{
		\begin{array}
			{ll}
			0 \text{ or } \loadshare , & \text{if } f_j^* = 1;\\
			\frac{\loadshare }{2} \text{ or } \loadshare , 
			& \text{if } 2 \le f_j^* \le \frac{N}{2}.
		\end{array}
		\right.
	\end{align}
\end{lemm}


\begin{proof}
	Suppose $\langle \vec{f}^*, \vec{s}^* \rangle$ is 
	an optimal solution to Problem (\ref{eqn:problem s F closed}). 
We claim that at
	least one scalar in $\vec{f}^*$ is equal to $N$, since otherwise,
	the total arrival rate is at most $\frac{1}{N^2} \cdot \frac{N}{2} \cdot N = \frac{1}{2}$ by the 
stripe size determination rule  (\ref{eqn:stripe size rule}), contradicting our assumption that
$\rho \ge \frac{2}{3} + \frac{1}{3N^2} > \frac{1}{2}$.
	In addition, if more than one scalars in $\vec{f}^*$ are equal to $N$, we add up the values of all these
scalars in $\vec{s}$, assign the total to one of them (say $s_m$), and zero out the rest.
	This aggregation operation will not change the total arrival rate to any intermediate port because
	an input VOQ with stripe size equal to $N$ will spread its traffic evenly to all the intermediate 
	ports anyway.  


	
	If the optimal solution $\langle \vec{f}^*, \vec{s}^* \rangle$ does not satisfy 
	Property (\ref{eqn:extreme property}), we show that we can find another 
	optimal solution satisfying
	Property (\ref{eqn:extreme property}) through the following
	optimization problem.  
	Note that now with $\vec{f}^*$ fixed, $\mathbb{E}[
		\myexp{\theta X(\vec{s}, \vec{f}^*)}]$ is a function of only $\vec{s}$.
	\begin{align}
		\max_{\vec{s}}:&\ 
		\mathbb{E}[
		\myexp{\theta X(\vec{s}, \vec{f}^*)}]
		\label{eqn:extreme value prob obj}\\
		\text{s.t.:}& 0 \le s_j \le \loadshare, 
		\forall j \text{ with } f^*_j = 1;
		\label{eqn:extreme value prob constr 1}\\
		&\ \frac{\loadshare}{2} \le s_j \le \loadshare, 
		\forall j \text{ with } 2 \le f^*_j \le N/2;
		\label{eqn:extreme value prob constr 2 n/2}\\
		&\ \loadshare/2 \le s_m;
		\label{eqn:extreme value prob constr n}\\
		&\ \sum_{j} s_j f^*_j = \rho.
		\label{eqn:extreme value prob constr total}
	\end{align}

	$\mathbb{E}[\myexp{\theta X(\vec{s}, \vec{f}^*)}] = \frac{1}{N!} \sum_{\sigma}
	\exp\left(\theta\sum_{j = 1}^N s_j 
	\mathbbm{1}_{\{f^*_j \ge \sigma(j)\}}\right)$  where 
	$\sigma(j)$ is the index of the intermediate port that VOQ $j$
	is mapped to. 
Given $\vec{f}^*$, 
the objective function 
	$\mathbb{E}[\myexp{\theta X(\vec{s}, \vec{f}^*)}]$ is clearly
	a convex function of $\vec{s}$. 

	The feasible region of $\vec{s}$ defined by 
	(\ref{eqn:extreme value prob constr 1}) - (\ref{eqn:extreme value prob constr total}) is
	a (convex) polytope.  By the convex optimization theory~\cite{boyd2004convex}, the (convex) objective function
	reaches its optimum at one of the extreme points of the polytope.  
	The assumption $\rho \ge \frac{2}{3} + \frac{1}{3N^2}$ implies that  
	Constraint (\ref{eqn:extreme value prob constr n}) cannot be
	satisfied with equality. Thus, for any extreme point of the polytope, 
	each of Constraint (\ref{eqn:extreme value prob constr 1}) and 
	(\ref{eqn:extreme value prob constr 2 n/2}) must be satisfied with equality on one side.
	Thus there exists at least one optimal solution 
	satisfying Property (\ref{eqn:extreme property}).
	This establishes the required result.
\end{proof}

\subsubsection{Exploiting Negative Associations}
\label{sec:Negative Association}

Suppose the optimization objective $\mathbb{E}[\myexp{\theta X(\vec{s}, \vec{f})}]$ is maximized at $\vec{s}^*$, an extreme point 
of the above-mentioned polytope, and $\vec{f}^*$.  
In this section, we 
prove that $\mathbb{E}[\myexp{\theta X(\vec{s}^*, \vec{f}^*)}]$ 
is upper-bounded by the right-hand side (RHS) of the inequality in Theorem 2, by exploiting the negative associations -- induced by $\vec{s}^*$ being an 
extreme point -- among 
the random variables that add up to $X(\vec{s}^*, \vec{f}^*)$.  For notational convenience, we drop the asterisk from 
$\vec{s}^*$ and $\vec{f}^*$, and simply write them as $\vec{s}$ and $\vec{f}$.  



Again for notational convenience, we further drop the terms 
$\vec{f}$ and $\vec{s}$ from $X(\vec{s}, \vec{f})$, and simply write it as $X$.
Such a dependency should be clear from the context.
We now split $X$ into two random variables $X^L$ and $X^U$ and a constant $X^d$, by splitting 
each $s_\ell$ into $s^L_\ell$ and $s^U_\ell$, and each $f_\ell$ into $f^L_\ell$ and $f^U_\ell$, as follows.
If $f_\ell = N$, we let $s^L_\ell = s^U_\ell = 0$. Otherwise,
we define $s^L_\ell = \min\{s_\ell, \alpha/2\}$ and 
$s^U_\ell = \max \{0, s_\ell - \alpha/2 \}$. 
We then define 
$f^L_\ell = f_\ell \mathbbm{1}_{\{s^L_\ell = \alpha/2\}}$ and 
$f^U_\ell = f_\ell \mathbbm{1}_{\{s^U_\ell = \alpha/2\}}$.
It is not hard to verify that $s^L_\ell + s^U_\ell = s_\ell$ and $f^L_\ell + f^U_\ell = f_\ell$, 
for $\forall \ell$ if $f_\ell \le N/2$.
We further define $X^L = \sum_{\ell = 1}^N 
\alpha/2 \cdot  \mathbbm{1}_{\{f^L_{\sigma^{-1}(\ell)} \ge \ell\}}$
and $X^U = \sum_{\ell = 1}^N 
\alpha/2 \cdot  \mathbbm{1}_{\{f^U_{\sigma^{-1}(\ell)} \ge \ell\}}$.
Finally, we define $X^d = \sum_{\ell=1}^N s_\ell \mathbbm{1}_{\{f_\ell = N\}}$.
$X^d$ is the total arrival rate from the (high-rate) VOQs that have 
stripe size $N$, which is a constant  because how much traffic any of these VOQs sends to 
the queue is not affected by which primary intermediate port $\sigma$ maps the VOQ to.  
It is not hard to verify that $X = X^d + X^U + X^L$.






By the Cauchy-Schwartz inequality, we have
\begin{align}
\nonumber
	\mathbb{E}[\exp(\theta X(\vec{r}))] &=
	\exp(\theta X^d(\vec{r})) \mathbb{E}
	[\exp(\theta X^U(\vec{r})) \exp(\theta X^L(\vec{r})) ]\\
	&\le 
	\exp(\theta X^d(\vec{r})) \left(\mathbb{E}[\exp(2\theta X^L(\vec{r}))]\right)^{1/2}\nonumber\\
	&\quad \cdot \left(\mathbb{E}[\exp(2\theta X^U(\vec{r}))]\right)^{1/2}.
	\label{eqn:Cauchy-Schwartz bound}
\end{align}


Let $p^L_\ell = \mathbb{P}(f^L_{\sigma^{-1}(\ell)} \ge \ell)$.  We now upper bound the MGF of $X^L$ as follows:
\begingroup
\allowdisplaybreaks
\begin{align}
	\mathbb{E}[\exp(2\theta X^L)] &= 
	\mathbb{E}\left[\exp\left(2\theta \cdot \alpha/2 \cdot
	\sum_{\ell=1}^{N} \mathbbm{1}_{\{f^L_{\sigma^{-1}(\ell)}\ge \ell\}} \right)\right]\nonumber\\
	&= 
	\mathbb{E}\left[ 
		\prod_{\ell = 1}^{N} 
		\exp\left(\theta \alpha \mathbbm{1}_{\{f^L_{\sigma^{-1}(\ell)}\ge \ell\}}\right)
	\right]\nonumber\\
	&\le 
	\prod_{\ell = 1}^{N} 
	\mathbb{E}\left[\exp\left(\theta \alpha \mathbbm{1}_{\{f^L_{\sigma^{-1}(\ell)} \ge \ell\}}\right)\right]\nonumber\\
	&=
	\prod_{\ell = 1}^{N/2} 
	\mathbb{E}\left[\exp\left(\theta \alpha \mathbbm{1}_{\{f^L_{\sigma^{-1}(\ell)} \ge \ell\}}\right)\right]\nonumber\\
	&\le 
	\prod_{\ell = 1}^{N/2} h(p^*(\theta \alpha), \theta \alpha)
	\exp(\theta \alpha p^L_\ell)\nonumber \\
	&=  (h(p^*(\theta \alpha), \theta \alpha))^{N/2}
	\prod_{\ell = 1}^{N/2}
	\exp(\theta \alpha p^L_\ell).	\label{eqn:bound on L}
\end{align}
\endgroup
The first inequality holds due to Lemma \ref{prty:na to product form} and the following two facts.  First, 
$\{f^L_{\sigma^{-1}(1)}, f^L_{\sigma^{-1}(2)}, \cdots, f^L_{\sigma^{-1}(N)}\}$,
as a uniform random permutation of $\{f^L_1, f^L_2, \cdots, f^L_N\}$, are negatively associated
according to 
Lemma~\ref{lemma:permutation is NA}.  Second,  $\exp(\theta \alpha\mathbbm{1}_{\{x \ge \ell\}})$ is a nondecreasing function of $x$ for any given $\ell$
when $\theta >0$.  The third equality holds because with the way we define $f^L$ above, $f^L_{\sigma^{-1}(\ell)} \le N/2$ for any $\ell$.
The last inequality holds because each term in the product on the left-hand side, $\mathbb{E}[\exp(\theta \alpha \mathbbm{1}_{\{f^L_{\sigma^{-1}(\ell)} \ge \ell\}})]$,
is upper bounded by the corresponding term in the product on the RHS, $h(p^*(\theta \alpha), \theta \alpha)\exp(\theta \alpha p^L_\ell)$.  It is not hard to verify 
because each $\mathbb{E}[\exp(\theta \alpha \mathbbm{1}_{\{f^L_{\sigma^{-1}(\ell)} \ge \ell\}})]$ is the MGF of a Bernoulli random variable scaled by a constant
factor $\theta \alpha$, and the function $h$ reaches its maximum at $p^*$ as defined in Theorem \ref{theo:overflow bound mgf}.

Letting
$p^U_\ell = \mathbb{P}(f^U_{\sigma^{-1}(\ell)} \ge \ell)$,
we can similarly bound $\mathbb{E}[\exp(2\theta X^U)]$ as follows:
\begin{align}
\label{eqn:bound on U}
	\mathbb{E}[\exp(2\theta X^U)] \le
	(h(p^*(\theta \alpha), \theta \alpha))^{N/2}
	\prod_{\ell = 1}^{N/2} 
	\exp(\theta \alpha p^U_\ell).
\end{align}

Combining (\ref{eqn:Cauchy-Schwartz bound}), (\ref{eqn:bound on L}),
and (\ref{eqn:bound on U}), 
we obtain
\begin{align}
\label{eqn:bound before centering}
	\mathbb{E}[\exp(\theta X)]
	\le &
	(h(p^*(\theta \alpha), \theta \alpha))^{N/2}\nonumber \\
	&\exp\left(\theta 
	(X^d + \frac{\alpha}{2}\sum_{\ell=1}^{N/2} (p^L_\ell + p^U_\ell))\right)\nonumber\\
	&= 
	(h(p^*(\theta \alpha), \theta \alpha))^{N/2}
	\exp\left(\theta\rho/N\right).\nonumber
\end{align}
The final equality holds since $\rho/N = \mathbb{E}[X]   = \mathbb{E}[X^d + X^L + X^U] = 
X^d + \frac{\alpha}{2}\sum_{\ell=1}^{N/2} (p^L_\ell + p^U_\ell))$.
This concludes the proof of the theorem.

\section{Expected delay at intermediate stage}
\label{sec:delay at intermediate}

In this section, we analyze the expected queue length at the intermediate ports. 
We are interested in this metric for two reasons. First, the 
queue length at these stations contribute to the total delay of the packets
in this switch. Secondly and
more importantly, when the switch detects changes of arrival 
rates to the input VOQs, 
the frame sizes of the input VOQs need to be redesigned. 
Before the packets with the new frame sizes are spread to the 
intermediate ports, the switch needs to make sure 
that the packets of previous frame sizes are all cleared. 
Otherwise, there could be reordering of the packets. 
The expected duration of this clearance
phase is the average queue length at the stations.

We assume that the arrivals in different periods are identical
and independently distributed. For a queuing system, given
that the arrival rate is smaller than the service rate, the
queue is stable and the queue length is mainly due to the 
burstiness of the arrivals. To make the arrivals have the maximum 
burstiness, we assume the arrival in each cycle ($N$ time slots) 
is a Bernoulli random variable that
takes value $N$ with probability $\frac{\rho}{N}$ and 
0 with probability $1 - \frac{\rho}{N}$.
We can construct a discrete time 
Markov process for the queue length \emph{at the end of each cycle}. Its state space is 
$\{0, 1, \cdots, \}$ and the nonzero transition probabilities are 
$\mathbb{P}_{i, i+N-1} = 1 - \frac{\rho}{N}$ for $i \ge 0$ and 
$\mathbb{P}_{i, i-1} = \frac{\rho}{N}$ for $i \ge 1$,
and $\mathbb{P}_{0, 0} = \frac{\rho}{N}$. 

Given the transition matrix, we can solve for the stationary distribution, from which
we can obtain the expected queue length. As figure \ref{fig:delay} shows, the average delay in units of 
period is linear in the switch size. Given that our delay is analyzed against the worst burstiness
and one cycle is at the magnitude of millisecond, this 
delay is very acceptable. 

\begin{figure}[ht]
  \centering
    \includegraphics[width=0.5\textwidth]{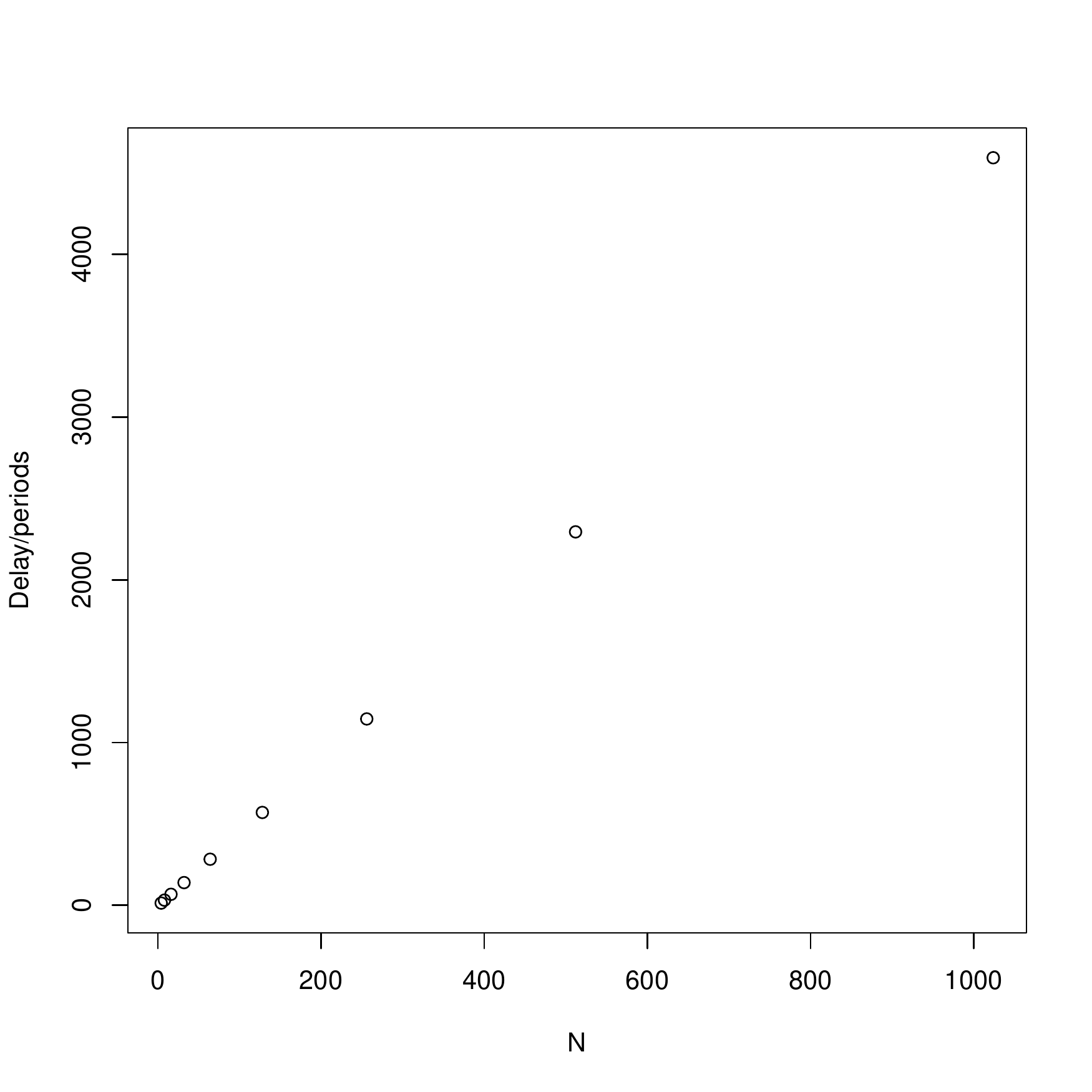}
    \caption{Expected delay when $\rho = 0.9$}
\label{fig:delay}
\end{figure}

\section{Simulation}
\label{sec:simu}

In this section, we present the simulation results of our switching algorithm
and other existing switching algorithms 
under different traffic scenarios. 
The switching algorithms studied in this section include the baseline load-balanced
switch~\cite{chang2002load}, the uniform frame spreading 
(UFS) switch~\cite{keslassy2004load}, the
full-ordered frame first (FOFF) switch~\cite{keslassy2004load}, and
the padded frame (PF) switch~\cite{jaramillo2008padded}. The baseline load-balanced 
switch does not guarantee packet ordering, but it provides the lower bound of
the average delay that a load-balanced switch can achieve. UFS, FOFF, and PF 
are known to provide reasonably good performance and all of them guarantee 
packet ordering.

In our simulation experiments, we assume a Bernoulli arrival to 
each input port, i.e. in each time slot, there is a packet arrival
with probability $\rho$. The $\rho$ parameter will be varied to study the performance 
of the switches under different loads. We also use different ways to 
determine the output port of each arrival packet. The size of the switch in the simulation
study is $N = 32$.

\begin{figure}[htb]
  \centering
    \includegraphics[width=0.45\textwidth]{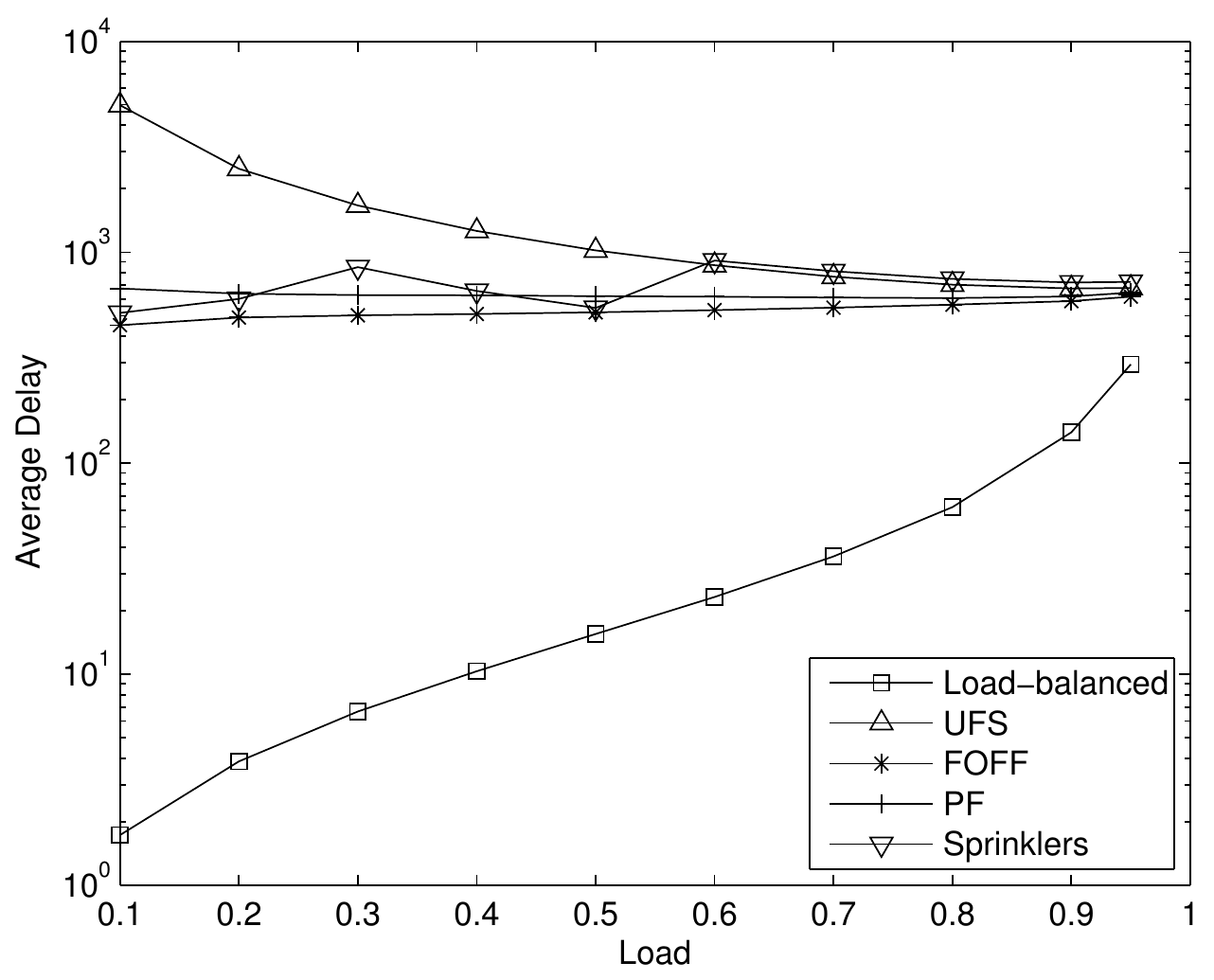}
    \caption{Average delay under uniform traffic}
\label{fig:uniform simulation}
\end{figure} 

\begin{figure}[htb]
  \centering
    \includegraphics[width=0.45\textwidth]{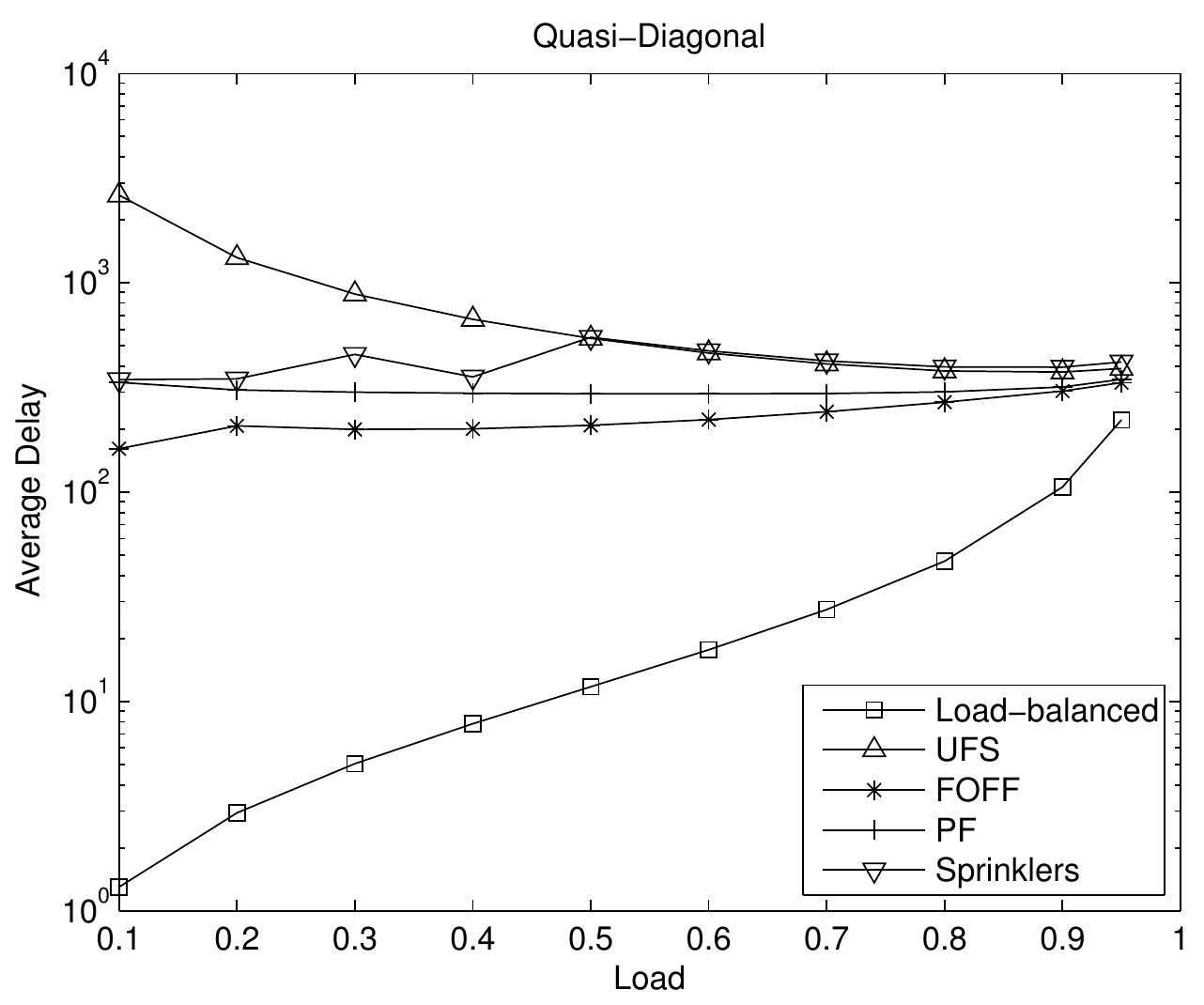}
    \caption{Average delay under diagonal traffic}
\label{fig:diagonal simulation}
\end{figure} 

Our first set of experiments assume uniform distribution of the destination port
for the arrival packets, i.e. a new packet goes to output $j$ 
with probability $\frac{1}{N}$. The simulation results are as shown in 
Fig.~\ref{fig:uniform simulation}. The second set of experiments assume a diagonal distribution. A new
packet arriving at input port $i$ goes to output $j = i$ with probability $\frac{1}{2}$,
and goes to any other output port with probability $\frac{1}{2(N-1)}$. 
The results are as shown in Fig.~\ref{fig:diagonal simulation}.

Our first observation of the comparison is that, 
under both the uniform and the diagonal traffic patterns, 
compared to UFS, our switch significantly reduces the average
delay when the traffic load is low. This is understandable because
the switch does not need to wait for additional packet to form a 
full frame. Secondly, like PF and FOFF, the average delay of our switching algorithm 
is quite stable under different traffic intensities. This property is 
quite preferable in practice. Finally, our switch has similar delay performance with
PF and FOFF while the implementation of our switch is simpler than the other two switches.


\section{Conclusion}
In this paper, we proposed Sprinklers,
a simple load-balanced switch architecture based on the combination of randomization and variable-size striping.
Sprinklers has comparable implementation cost and performance as the baseline load-balanced switch, but yet can guarantee packet ordering.
We rigorously proved using worst-case large deviation techniques that Sprinklers can achieve near-perfect load-balancing under arbitrary admissible traffic.
Finally, we believe that this work will serve as 
a catalyst to a rich family of solutions based on the simple principles of 
randomization and variable-size striping.
\label{sec:conc}

\bibliographystyle{plainnat}
\bibliography{sprinklers}


\end{document}